\documentclass[journal,10pt]{IEEEtran} \newcommand*{\COLS}{}%

\usepackage{graphicx}
\usepackage{epsfig}
\usepackage{epstopdf}
\usepackage{amsmath,amsthm,amsfonts,amssymb}
\usepackage{mathabx}
\usepackage{bbm}
\usepackage{multirow}
\usepackage{cite}
\usepackage{tikz}
\usepackage{tikz-cd}
\usepackage{ellipsis}
\usepackage{xcolor}
\usepackage{xifthen}
\usepackage{subcaption}
\usepackage{cancel}
\usetikzlibrary{patterns,arrows.meta}
\usepackage[linesnumbered, ruled, vlined, noend]{algorithm2e}

\newtheorem{theorem}{Theorem}
\theoremstyle{definition}
\newtheorem{lemma}{Lemma}
\theoremstyle{definition}

\theoremstyle{definition}
\newtheorem{definition}{Definition}
\theoremstyle{definition}

\theoremstyle{remark}

\newcommand{\realnumber}{\mathbb{R}}

\DeclareMathOperator{\bpsk}{bpsk}
\DeclareMathOperator{\scdec}{\mathsf{SC}}
\DeclareMathOperator{\ascdec}{\mathsf{aSC}}
\DeclareMathOperator{\dec}{dec}
\DeclareMathOperator{\adec}{adec}

\DeclareMathOperator{\eval}{eval}
\DeclareMathOperator{\sgn}{sgn}
\newcommand{\Fbin}{\mathbb{F}_2}
\newcommand{\Zbin}{\mathbb{Z}_{2^n}}

\newcommand{\Iset}{\mathcal{I}}
\newcommand{\Fset}{\mathcal{F}}

\newcommand{\code}[1]{\mathcal{#1}}

\newcommand{\T}{\mathsf{T}}
\newcommand{\GA}{\mathsf{GA}}

\newcommand{\GL}{\mathsf{GL}}
\newcommand{\LTA}{\mathsf{LTA}}
\newcommand{\LTL}{\mathsf{LTL}}
\newcommand{\BLTA}{\mathsf{BLTA}}
\newcommand{\BLTL}{\mathsf{BLTL}}
\newcommand{\OBLT}{\mathsf{OBLT}}
\newcommand{\UTL}{\mathsf{UTL}}
\newcommand{\PL}{\mathsf{PL}}
\newcommand{\Aut}{\mathsf{Aut}}
\newcommand{\pid}{\mathbbm{1}}
\newcommand{\bin}[1]{\widehat{#1}}
\newcommand{\invbin}[1]{\widecheck{#1}}
\newcommand{\elmat}[2]{E_{(#1,#2)}}
\newcommand{\elperm}[2]{\epsilon_{(#1,#2)}}
\newcommand{\EC}{\mathsf{EC}}

\begin{document}

\title{Group Properties of Polar Codes for \\ Automorphism Ensemble Decoding}


\author{\IEEEauthorblockN{Valerio Bioglio, Ingmar Land, Charles Pillet} 
\thanks{Ingmar Land was with the Mathematical and Algorithmic Sciences Lab, Paris Research Center, Huawei Technologies Co. Ltd. 92100 Boulogne-Billancourt, France, when he provided his contribution to the paper. He is now with Infinera Corporation, Ottawa, ON K2K 2X3, Canada (email: ingmar.land@ieee.com). 
Charles Pillet and Valerio Bioglio are with the Algorithmic and Mathematical Science Lab, Paris Research Center, Huawei Technologies France S.A.S.U. (email: charles.pillet1@huawei.com; valerio.bioglio@huawei.com).}}
\maketitle

\begin{abstract}
In this paper, we propose an analysis of the automorphism group of polar codes, with the scope of designing codes tailored for \emph{automorphism ensemble} (AE) decoding. 
We prove the equivalence between the notion of \emph{decreasing monomial codes} and the universal partial order (UPO) framework for the description of polar codes. 
Then, we analyze the algebraic properties of the \emph{affine automorphisms group} of polar codes, providing a novel description of its structure and proposing a classification of automorphisms providing the same results under permutation decoding. 
Finally, we propose a method to list all the automorphisms that may lead to different candidates under AE decoding; by introducing the concept of \emph{redundant} automorphisms, we find the maximum number of permutations providing possibly different codeword candidates under AE-SC, proposing a method to list all of them. 
A numerical analysis of the error correction performance of AE algorithm for the decoding of polar codes concludes the paper.
\end{abstract}

%
%

\section{Introduction}
\label{sec:intro}

Polar codes \cite{ArikanFirst} are a class of linear block codes relying on the phenomenon of channel polarization.  
Under successive cancellation (SC) decoding, they can provably achieve capacity of binary memoryless symmetric channels for infinite block length.  
However, in the short length regime, the performance of polar codes under SC decoding is far from state-of-the-art channel codes.  
To improve their error correction capabilities, the use of a list decoder based on SC scheduling, termed as SC list (SCL) decoding, has been proposed in \cite{TalSCL}. 
The concatenation of a cyclic redundancy check (CRC) code to the polar code permits to greatly improve its error correction performance \cite{CRCaidedSCL}, making the resulting CRC-aided SCL (CA-SCL) decoding algorithm the de-facto standard decoder for polar codes adopted in 5G standard \cite{polar_5G}. 
In order to avoid the extra decoding delay due to information exchange among parallel SC decoders in hardware implementation of CA-SCL decoders, permutation-based decoders have been proposed in \cite{PermGross} for SC; a similar approach was proposed in \cite{BPLRM} for belief propagation (BP) and in \cite{SCANL} for soft cancellation (SCAN). 
In a permutation-based decoder, $M$ instances of the same decoder are run in parallel on permuted factor graphs of the code, or alternatively on permuted versions of the received signal. 
Even if the decoding delay is reduced, the error correction performance gain of permutation-based decoders is quite poor due to the alteration of the frozen set caused by the permutation. 
Research towards permutations not altering the frozen set were carried out \cite{PermDecRussian}; such permutations are the \emph{automorphisms} of the code and form the group of permutations mapping a codeword into another codeword. 
This new decoding approach, i.e. the use of automorphisms in a permutation decoder, is referred to as \emph{automorphism ensemble} (AE) decoding.

Given their affinity with polar codes, the analysis of the automorphisms group of Reed-Muller (RM) codes is providing a guidance for this search. 
In fact, the automorphism group of binary Reed-Muller code is known to be the general affine group \cite{WilliamsSloane}, and an AE decoder for RM codes has been proposed in \cite{geiselhart2020automorphism} using BP as component decoder. 
The boolean code nature of RM codes permitted authors in \cite{BardetPolyPC} to propose a \emph{monomial code} description of polar codes. 
Through this definition, in the same papers the authors proved that the group of lower-triangular affine ($\LTA$) transformations form a subgroup of the automorphisms group of polar codes. 
However, it has been proved in \cite{geiselhart2020automorphism} that $\LTA$ transformations commute with SC decoding, in the sense that $\LTA$ transformations do not alter the result of SC decoding process; this property leads to no gain under AE decoding when these automorphisms are used. 
Fortunately, the $\LTA$ transformation group is not always the full automorphisms group of polar codes \cite{PC_UTL_design}; if carefully designed, a polar code exhibits a richer automorphisms group. 
In \cite{geiselhart2021automorphismPC,li2021complete}, the complete affine automorphisms group of a polar code is proved to be the block-lower-triangular affine ($\BLTA$) group.
Authors in \cite{geiselhart2021automorphismPC,PC_UTL_design} showed that automorphisms in $\BLTA$ group (not belonging to $\LTA$) can be successfully used in an AE decoder; in practice, these automorphisms are not absorbed by the SC decoder. 
Unfortunately, the group of affine automorphisms of polar codes asymptotically is not much bigger than $\LTA$ \cite{not_many_autos}. 
These discoveries paved the way for the analysis of the automorphism group of polar codes, and in general on the struggle of constructing polar codes having "good" automorphism groups, i.e. a set of automorphisms to be used in an AE decoder. 
Author in \cite{Perm_Russian_polar_subcode} followed another approach, focusing on the use of polar subcodes in conjunction with a peculiar choice of the permutation set to design polar codes for AE decoding. 
However, this approach is less systematic than the analysis of the automorphism group of the polar code, leading to less predictable gains. 

In this paper, we propose an analysis of the automorphisms group of polar codes, with the scope of designing codes exhibiting good error correction performance under AE decoding. 
To begin with, in Section~\ref{sec:monomial-codes} we delve into the monomial codes description of polar codes proposed in \cite{BardetPolyPC}, proving the equivalence between the notion of \emph{decreasing monomial codes} and the universal partial order (UPO) framework proposed in \cite{PartialOrder} for the description of polar codes.
This parallelism will allow us to analyze the algebraic properties of the \emph{affine automorphisms group} of polar codes in Section~\ref{sec:automorphisms}, namely a sub-group of the complete automorphisms group of polar codes whose elements can be described through affine transformations. 
Thanks to this analysis, we will provide a novel proof of the structure of the affine automorphism group of polar codes, that is different from the ones provided in \cite{geiselhart2021automorphismPC,li2021complete} and is connected to the UPO framework. 
Next, in Section~\ref{sec:eq_classes} we study the nature of AE decoders introducing the notion of \emph{decoder equivalence}, namely a classification of automorphism providing the same results under permutation decoding. 
Thanks to this new formalism, we prove that $\LTA$ is not always the complete SC absorption group, namely that a larger set of automorphisms may be absorbed under SC decoding. 
This result is an extension of our preliminary analysis of decoder equivalence published in \cite{AE_class}.
We provide an alternative proof of a very recent result presented in \cite{SC_invariant}, namely that the complete SC absorption group has a $\BLTA$ structure. 
Finally, we propose a method to list all the automorphisms that may lead to different candidates under permutation decoding; by introducing the concept of \emph{redundant} automorphisms, we find the maximum number of permutations providing possibly different codeword candidates under AE-SC. 
Our method permits to easily list all these automorphisms, greatly simplifying the search for automorphisms to be used in AE-SC decoding. 
All the results presented in the paper are correlated with examples to guide the reader in the process of constructing and using automorphisms for AE decoding of polar codes. 
Section~\ref{sec:nem_res}, including a numerical analysis of the performance of AE algorithm for the decoding of polar codes, concludes the paper.


%

\section{Polar codes and monomial codes}
\label{sec:monomial-codes}

In this section we analyze the parallelism between polar codes and decreasing monomial codes. 
Monomial codes have been introduced in \cite{BardetPolyPC} as a family of codes including polar and RM codes. 
In the same paper, decreasing monomial codes are introduced to provide a monomial description of polar codes. 
Through this description, authors in \cite{BardetPolyPC} were able to prove that lower triangular affine ($\LTA$) transformations form a sub-group of the automorphisms group of polar codes. 
In this section, we extend the results in \cite{BardetPolyPC} proving the equivalence between the notion of \emph{decreasing monomial codes} and the universal partial order (UPO) framework proposed in \cite{PartialOrder}. 
UPO represents a method to partially sort virtual polarized channels independently of the actual channel used for the transmission. 
Due to the structure of the polarization phenomenon, some of the virtual channels are inherently better than others, and the UPO framework is able to catch this nuance and creates a structure of virtual channels. 
To the best of our knowledge, this is the first time that this equivalence is proved. 

\subsection{Polar Codes}
To begin with, we provide a definition of polar code. 
In the following, the binary field with elements $\{0,1\}$ is denoted by $\Fbin$, while the set of non-negative integers smaller than $N$ is written as $\mathbb{Z}_N = \{0,1,\ldots,N-1\}$.
\begin{definition}[Polar codes]
A polar code of length $N=2^n$ and dimension $K$ is defined by a transformation matrix $T_n = T_2^{\otimes n}$, where $T_2\triangleq [ \begin{smallmatrix} 1 & 0\\ 1 & 1 \end{smallmatrix} ]$, and a frozen set $\Fset \subset \Zbin$, or inversely by an information set $\Iset = \Fset^c = \Zbin  \backslash \Fset$. 
Encoding is performed as $x = u \cdot T_N$, where  $u_{\Fset} = 0$; the code is then given by 
\begin{equation}
	\code{C} = \{ x = u \cdot T_N : u \in \Fbin^N , u_\Fset = 0 \} .
\end{equation}
\end{definition}
Elements of the information set $\Iset$ are usually chosen according to their reliability; in practice, virtual channels originated by the transformation matrix are sorted in reliability order, and the indices of the $K$ most reliable ones form the information set of the code. 
The reliability of virtual channels depend on the characteristics of the transmission channel, forcing the code designer to recalculate the reliability order every time the channel parameters change. 
However, even if the polarization effect is non-universal, the structure of transformation matrix $T_N$ permits to define an \emph{universal partial order} (UPO) among virtual channels. 
In practice, we say that $i \preceq j$ if virtual channel represented by index $i$ is always weaker (i.e. less reliable) than the channel represented by index $j$, independently of the original channel \cite{PartialOrder}. 
A polar code fulfills the UPO if this partial reliability order is followed. 
\begin{definition}[Universal partial order]
A polar code is said to fulfill the universal partial order (UPO) if $\forall i,j \in \Zbin$ such that $i \preceq j$ then $i \in \mathcal{I} \Rightarrow j \in \mathcal{I}$.
\end{definition}
\begin{figure*}[tb]
	\begin{center}
		\resizebox{.98\textwidth}{!}{\input{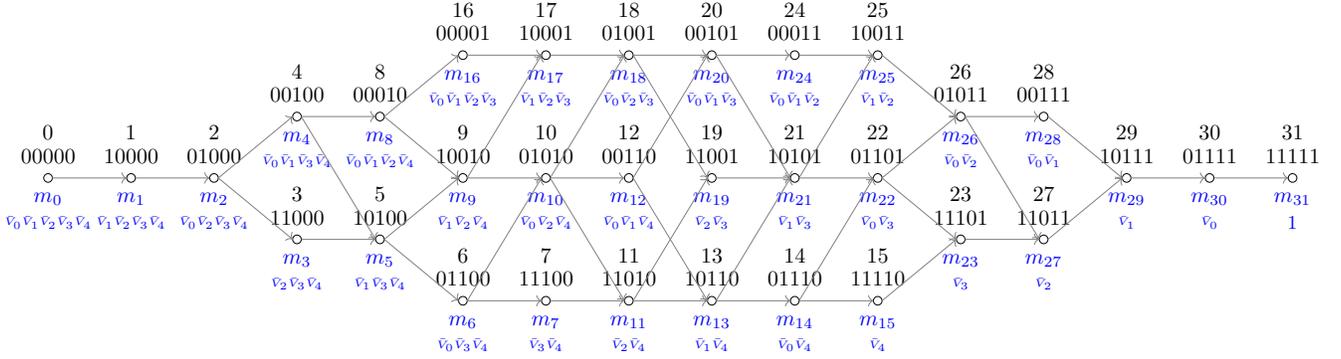}}
	\end{center}
	\caption{Hasse diagram of UPO for $N=32$; integers and their binary expansions are depicted in black, corresponding monomials are depicted in blue.}
	\label{fig:hasse}
\end{figure*}
Figure~\ref{fig:hasse} provides a graphical representation of the relations among the virtual channels for a polar code of length $N = 32$. 
A node in the diagram represents a virtual channel, while a directed edge represents the order relation between two virtual channels. 
Nodes that are connected by a directed path can be compared through the UPO, while the order relation between unconnected nodes depends on the transmission channel and should be verified. 
As an example, since nodes $5$ and $9$ are connected, we know that $5 \preceq 9$ and virtual channel $5$ is always weaker than virtual channel $9$; on the other hand, nodes $6$ and $9$ are not connected, and hence the order between these two virtual channels depends on the original channels. 
A polar code fulfilling the UPO can be described through few generators composing the minimum information set $\mathcal{I}_{min}$ \cite{geiselhart2021automorphismPC} as
\begin{equation}
\mathcal{I}=\bigcup_{j\in\mathcal{I}_{min}} \{i\in[N],j\preceq i\}.
\end{equation}
It is worth noticing that a polar code designed to be decoded under SCL may not fulfill the UPO; the universal reliability sequence standardized for 5G is an example of this exception \cite{polar_5G}. 
However, in the following we will focus on polar codes fulfilling the UPO, since we are interested in AE decoders having SC component decoders. 

The UPO is induced by noticing that the binary expansion of the virtual channel index actually provides the sequence of channel modifications, where a $0$ represents a channel degradation and a $1$ represents a channel upgrade. 
\begin{definition}[Binary expansion]
\label{def:bin_exp}
The \emph{canonical binary expansion} connecting $\Fbin^n$ and $\Zbin$,
\begin{center}
\begin{tikzcd}
v \in \Fbin^n \arrow[r, bend left=50,"\bin{\cdot}"]
\arrow[r, bend right=50,leftarrow,"\invbin{\cdot}"]
& \bin{v} \in \Zbin
\end{tikzcd}
\end{center}
maps an integer $v$ to its binary expansion $\bin{v}$ with least significant bit (LSB) on the left, namely with $v_0$ representing the LSB of $v$; zero padding is performed if the number of bits representing $v$ is smaller than $n$, namely if $v < 2^{n-1}$.
\end{definition}
So, given two virtual channel indices $i,j \in \Zbin$ with $i<j$, if $\bin{i}$ and $\bin{j}$ represent their binary expansions, we have that the corresponding nodes in the Hasse diagram are directly connected, and hence $i \preceq j$, when:
\begin{itemize}
\item if $\bin{i}$ and $\bin{j}$ only differ for a single entry $t$, where $\bin{i}_{t}=0$ and $\bin{j}_{t}=1$; 
\item if $\bin{i}$ and $\bin{j}$ only differ for two consecutive entries $t$ and $t+1$, where $\bin{i}_{t}=1,\bin{i}_{t+i}=0$ and $\bin{j}_{t}=0,\bin{j}_{t+1}=1$. 
\end{itemize}
These rules are iterated to generate the UPO among the virtual channels, and hence among the integers representing these channels. 

\subsection{Monomial Codes}
Monomial codes are a family of codes of length $N=2^n$ defined by evaluations of boolean functions of $n$ variables. 
A boolean function can be described as a map from a string of $n$ bits to a single bit defined by a specific truth table; this map can always be rewritten as a polynomial $f$ in the polynomial ring in $n$ indeterminates over $\Fbin$, i.e. $f \in \Fbin[V_0,\dots,V_{n-1}] \triangleq \Fbin^{[n]}$.
Since an element of $\Fbin^n$ can be interpreted as the binary representation of an integer, boolean functions can be seen as maps associating a bit to each integer in $\Zbin. $ 
Every boolean function can be written as a linear combination of monomials in variables $V_i$, forming a basis for the space of boolean functions. 
A monomial in $V_i$ is a product of powers of variables $V_i$, with $0\leq i < n$, having non-negative exponents, e.g. $V_0^2 V_2$. 
However, since $V_i^2=V_i$ in $\Fbin^{[n]}$, variables either appear or not in monomials. 
Similarly, \emph{negative monomials}, namely monomials in $\bar{V}_i = \neg V_i = (1 \oplus V_i)$, form a basis of the same space; in the following, we will use negative monomials to smoothly map polar codes to monomial codes, and we call $\mathcal{M}^{[n]} \subset \Fbin^{[n]}$ the set of monomials in $n$ indeterminates over $\Fbin$. 
Given their importance in the definition of boolean functions, we introduce here a notation to connect each monomial in $\mathcal{M}^{[n]}$ to an integer in $\Zbin$.
\begin{definition}[Monomials canonical map]
\label{def:mon_map}
The \emph{canonical map} connecting $\Zbin$ and $\mathcal{M}^{[n]}$,
\begin{align*}
	            \Zbin & \leftrightarrow \mathcal{M}^{[n]}  \\
	            t     & \leftrightarrow m_t  ,
\end{align*}
connects integer $t \in \Zbin$, where for some $Q \subset \mathbb{Z}_{n}$
\begin{equation}
t=\sum_{i \in Q} 2^i,
\end{equation}
to monomial $m_t \in \mathcal{M}^{[n]}$, defined by 
\begin{equation}
m_t = \prod_{i \notin Q} \bar{V}_i
\end{equation}
\end{definition}
In practice, the monomials canonical map transforms a monomial into an integer having zeroes in its binary expansion corresponding to the variable indices, and ones elsewhere.  
Equipped with this map, we can now define the monomial evaluation function. 
\begin{definition}[Evaluation function]
\label{def:ev_fun}
The evaluation function $\eval : \Fbin^{[n]} \rightarrow \Fbin^N$ checks the output of $f$ for all elements of $\mathbb{F}_2^n$ in increasing order; in practice, 
\begin{equation}
x^{(f)} = \eval(f) \triangleq (f(\bin{0}),f(\bin{1}),\ldots,f(\bin{N-1})).
\end{equation}
\end{definition}
As a consequence, every boolean function $f$ can be naturally associated to a binary vector $x^{(f)}$ of length $N$ through the evaluate function $\eval(f)$.
An example of this construction can be found in Appendix~\ref{app:eval}. 
Given that high degree monomials are defined as products of degree-one monomials, namely single variables, knowing the evaluations of single variables it is easy to calculate the evaluation of high-degree monomials as xor of the associated binary strings. 
This property permits to generate a family of block codes based on the evaluation of monomial functions. 
\begin{definition}[Monomial codes]
A \emph{monomial code} of length $N=2^n$ and dimension $K$ is given by the evaluations of linear combinations of the monomials included in a generating monomial set $\mathcal{G} \subset \mathcal{M}^{[n]}$, with $|\mathcal{G}|=K$.
\end{definition}
Given $n$ variables, it is possible to form $2^n$ different monomials, of which $\binom{n}{r}$ are of degree $r$. 
A monomial code of length $N=2^n$ and dimension $K$ is generated by picking $K$ monomials out of the $N$. 
Reed-Muller codes are monomial codes defined by picking all monomials up to a certain degree: in particular, $\mathcal{R}(r,n)$ code is generated by all monomials of degree smaller or equal to $r$. 
The family of monomial codes being so rich, in order to introduce a parallelism between these codes and the polar codes, a notion of partial order among monomials is needed. 
Authors in \cite{BardetPolyPC} propose such a partial order. 
Given two monomials $m_{t_1},m_{t_2} \in \mathcal{G}$ of degrees $s_1$ and $s_2$ respectively, where $m_{t_1}=\bar{V}_{i_0} \cdot \ldots \cdot \bar{V}_{i_{s_1-1}}$ and $m_{t_2}=\bar{V}_{j_0} \cdot \ldots \cdot \bar{V}_{j_{s_2-1}}$, we say that $m_{t_1} \preceq m_{t_2}$: 
\begin{itemize}
\item when $s_1=s_2=s$, then if and only if $i_l \leq j_l$ for all $l=0,\dots,s-1$.
\item when $s_1<s_2$, then if and only if there exists a monomial $m_{t'}$ such that $m_{t'} | m_{t_2}$, $deg(m_{t_1}) = deg(m_{t'})$ and $m_{t_1} \preceq m_{t'}$. 
\end{itemize}
It is worth noticing that we used the same symbol to define the partial order relations for integers and for monomials; in fact, we will see that these two orders are equivalent. 
Equipped with the notion of partial order among monomials, authors in \cite{BardetPolyPC} define the family of \emph{decreasing monomial codes} as the monomial codes for which for every monomial in $\mathcal{G}$, all its sub-monomial factors are also included in $\mathcal{G}$. 
\begin{definition}[Decreasing monomial codes]
	A monomial code is \emph{decreasing} if for every monomial $m_{t_1},m_{t_2} \in \mathcal{M}^{[n]}$ such that $m_{t_1} \preceq m_{t_2}$ then $m_{t_2} \in \mathcal{G} \Rightarrow m_{t_1} \in \mathcal{G}$.
\end{definition}

\subsection{Monomial and Polar Codes Equivalence}
\begin{table}[ht]
\caption{Canonical map between integers and monomials for $n=3$.} 
\centering 
	\begin{tabular}{|r|r|c|c|c|}
		\hline
		degree & monomial & evaluation & row of $T_8$ & expansion\\
		\hline
		$0$ & $m_7=1$      & $11111111$ & 7 & 111 \\
		\hline
		\multirow{3}{*}{$1$} & $m_6=\bar{V}_0$    & $10101010$ & 6 & 011 \\
		& $m_5=\bar{V}_1$ & $11001100$ & 5 & 101 \\
		& $m_3=\bar{V}_2$ & $11110000$ & 3 & 110 \\
		\hline
		\multirow{3}{*}{$2$} & $m_4=\bar{V}_0\bar{V}_1$ & $10001000$ & 4 & 001 \\
		& $m_2=\bar{V}_0\bar{V}_2$ & $10100000$ & 2 & 010 \\
		& $m_1=\bar{V}_1\bar{V}_2$ & $11000000$ & 1 & 100 \\
		\hline
		$3$ & $m_0=\bar{V}_0\bar{V}_1\bar{V}_2$ & $10000000$ & 0 & 000 \\
		\hline
	\end{tabular}
	\label{tab:pol2mon}
\end{table}
Polar codes can be described as monomial codes. 
In fact, the kernel matrix $T_2$ can be seen as the evaluation of monomials in $\Fbin[\bar{V}_0]$; $\eval(\bar{V}_0) = [1,0]$ represents the first row of $T_2$, while $\eval(1) = [1,1]$ represents its second row. 
An example of this parallelism is provided in Table~\ref{tab:pol2mon} for $n=3$. 
This parallelism can be extended to polar codes of any length $N=2^n$, such that each monomial in $\mathcal{M}^{[n]}$ is represented by an integer in $\Zbin$ through the monomials canonical map. 
This map permits to connect row $t$ of transformation matrix $T = T_2^{\otimes n}$ to a unique monomial $m_t$; 
Then, polar codes can be equivalently defined in terms of information set $\mathcal{I}$ or generating monomial set $\mathcal{G}$, since $m_t \in \mathcal{G} \Leftrightarrow t \in \mathcal{I}$.  
To summarize, polar codes can be seen as monomial codes where the generating monomials are chosen according to polarization effect. 

In the following, we prove the main result of this section, namely that UPO property for polar codes is equivalent to decreasing property for monomial codes. 
As a consequence, every polar code fulfilling the UPO is also a decreasing monomial code, and vice versa. 
The fact that UPO is a sufficient condition for the generation of decreasing monomial codes has been proved in \cite{BardetPolyPC}, while, for the best of our knowledge, this is the first time that it is proven that it is also a necessary condition. 
Before proving the main theorem of the section, we need to prove the equivalence between the partial order defined on integers and the partial order defined on monomials. 
\begin{lemma}
\label{lem:RS_SD}
For every $a,b \in \Zbin$, then $b \preceq a \Leftrightarrow m_a \preceq m_b$
\begin{proof}
\textit{Sufficient condition:} 
Given $a,b \in \Zbin$ such that $b \preceq a$, we want to prove that $m_a \preceq m_b$ by checking if the two conditions for the partial ordering of the monomials are satisfied. 
If we call $HW(\bin{t})$ the Hamming weight of the binary expansion of integer $t$, then $HW(\bin{t}) = n-deg(m_{t})$ by Definition~\ref{def:mon_map}, and the degree conditions of monomial partial orders can be rewritten as conditions on the Hamming weights of binary expansions of integers. 

If $HW(\bin{a}) = HW(\bin{b}) = n-s$, then $deg(m_{a}) = deg(m_{b}) = s$; if we call $\bar{S}$ the set of variable indices composing monomial $m_t$, namely  $m_t = \prod_{i \in \bar{S}} \bar{V}_i \in \mathcal{M}^{[n]}$, then $m_a \preceq m_b$ if and only if $i^{(a)}_l \leq j^{(b)}_l$ for all $l=0,\dots,s-1$, where $i^{(t)}_l$ is the $l$-th bit in the binary expansion of integer $t$. 
Since $b \preceq a$ and $HW(\bin{a}) = HW(\bin{b})$, then it is possible to create a chain of integers $t_0,\dots,t_r$ such that $b = t_r \preceq t_{r-1} \preceq \ldots \preceq t_1 \preceq t_0 = a$ such that they all have the same Hamming weight and each couple of subsequent integers in the chain only differ for two consecutive entries. 
By definition, we have that $m_{t_{i}} \preceq m_{t_{i+1}}$ for the monomials partial order definition, and then the chain can be rewritten as $m_a = m_{t_0} \preceq m_{t_{1}} \preceq \ldots \preceq m_{t_{r-1}} \preceq m_{t_r} = m_b$.

Alternatively, if $HW(\bin{b}) < HW(\bin{a})$, than there is at least one chain of integers such that $b = t_r \preceq t_{r-1} \preceq \ldots \preceq t_1 \preceq t_0 = a$, where each couple of subsequent integers only differ for a single or two consecutive entries. 
Now, let us focus on three elements of the chain $t_{i+1} \preceq t_i \preceq t_{i-1}$ such that $t_{i+1}$ and $t_i$ differ for a single entry, while $t_{i}$ and $t_{i-1}$ differ for two consecutive entries; in this case, there exists another integer $t'_i$ such that $t_{i+1} \preceq t'_i \preceq t_{i-1}$ and $t_{i+1}$ and $t'_i$ differ for two consecutive entries, while $t'_{i}$ and $t_{i-1}$ differ for a single entry. 
In practice, it is always possible to invert the application of two different rules in the chain. 
As a consequence, it is always possible to create a chain $b = t'_r \preceq t'_{r-1} \preceq \ldots \preceq t'_1 \preceq t'_0 = a$ such that integers from $b = t'_r$ to a certain integer $t'_c$ differ for a single entry, while from $t'_c$ to $t'_0 = a$ two consecutive integers differ for two consecutive entries. 
Then, from this chain we extract element $t'_c$ with $b \preceq t'_c \preceq a$; by construction, the set of the indices of the positions of zeroes in the binary expansion of $t'_c$ is a subset of the same set of $b$, and hence $m_{t'_c} | m_b$. 
Moreover, $HW(t'_c) = HW(\bin{a})$, and then $deg(m_{t'_c}) = deg(m_a)$ and thus $m_a \preceq m_b$. 

\textit{Sufficient condition:} 
Given $m_a,m_b \in \mathcal{M}^{[n]}$ such that $m_a \preceq m_b$, we want to prove that $b \preceq a$ by checking if the two conditions for the partial ordering of the integers are satisfied. 
If $deg(m_{a}) = deg(m_{b})$, we can create a chain of intermediate monomials $m_a = m_{t_r} \preceq m_{t_{r-1}} \preceq \ldots \preceq m_{t_1} \preceq m_{t_0} = m_b$ such that they all have the same degree and each couple of subsequent monomials in the chain only differ by a variable; in other words, there exists a variables swap chain passing from $m_a$ to $m_b$ where each step of the chain can be sorted according to the partial order. 
This monomials chain can be mirrored to the corresponding integers chain $b = t_0 \preceq t_{1} \preceq \ldots \preceq t_{r-1} \preceq t_r = a$, where consecutive integers differ for two consecutive entries, and thus $b \preceq a$.
Alternatively, If $deg(m_{a}) < deg(m_{b})$, then there exists a monomial $m_{t}$ dividing $m_{b}$ and having the same degree of $m_{a}$ such that $m_{a} \preceq m_{t} \preceq m_{b}$. 
Then, $t \preceq a$ for the previous case, and $b \preceq t$. 
\end{proof}
\end{lemma}
Equipped with Lemma~\ref{lem:RS_SD}, we can now prove the main result of this section, namely the equivalence between UPO polar codes and decreasing monomial codes. 
\begin{theorem}
A polar code design is compliant with the UPO framework if and only if it is a decreasing monomial code. 
\begin{proof}
First, we assume that the information set of the polar code is compliant with the UPO framework. 
We need to prove that, if $m_t \in \mathcal{G}$, then also $m_{t'} \in \mathcal{G}$ for every $m_{t'} \preceq m_{t}$. 
According to Lemma~\ref{lem:RS_SD}, $t \preceq t'$, and since $t \in \mathcal{I}$ we have that also $t' \in \mathcal{I}$ for the UPO hypothesis, and hence $m_{t'} \in \mathcal{G}$.
Second, we assume the code to be decreasing monomial; now we need to prove that for every $t \in \mathcal{I}$, then also $m_{t'} \in \mathcal{G}$ for every $t \preceq t'$.  
Again, Lemma~\ref{lem:RS_SD} says that $m_{t'} \preceq m_{t}$, and since $m_t \in \mathcal{G}$ then also $m_{t'} \in \mathcal{G}$ for the decreasing monomial hypothesis, and hence $t' \in \mathcal{I}$.
\end{proof}
\end{theorem}


\section{Polar code automorphisms}
\label{sec:automorphisms}

Automorphisms are permutations of code bit positions that are invariant to the code, namely that map codewords into codewords. 
The analysis of the automorphism group of a code permits to discover hidden symmetries of the codewords, and can be used to find new properties of the code. 
In this paper, our study of the automorphism group of polar codes is driven by the will of improving AE decoding algorithms for this family of codes.  
In this section, we first revise permutations defined by general affine transforms and then discuss properties of such permutations that are automorphisms of polar codes. 
Moreover, we provide all the tools to help the reader to map an affine transformation to the related code bit permutation, by explicitly showing how to pass from one to the other. 
Equipped with this map, it will be easier for the reader to understand the main results of following sections and to reproduce the results presented in Section~\ref{sec:nem_res}. 

\subsection{Permutations as affine transformations}
\label{sec:}

\begin{definition}[Permutation]
	A permutation $\pi$ over the set $\Zbin$,
	\begin{align*}
		\pi : \Zbin &\rightarrow \Zbin   \\
		i &\mapsto     \pi(i)  ,
	\end{align*} 
	is a bijection of $\Zbin$ onto itself.
\end{definition}
The \emph{trivial (identity) permutation} is written as $\pid$ and maps every integer to itself. 
Permutations can be applied to vectors in different ways; in the following, for permutations of vectors we will use the \emph{functional passive notation}, where the element in position $i$ is replaced by element in position $\pi(i)$ after the permutation and permutations are concatenated giving priority to the right \cite{perm_book}. 
\begin{definition}[Vector permutation]
	\label{def:vector-permutation}
	Given a permutation $\pi$, the vector $y = ( y_0 , y_1 , \ldots , y_{2^n-1} )$ is called the permuted vector of vector $x = ( x_0 , x_1 , \ldots , x_{2^n-1} )$, if and only if
	\begin{equation*}
		y_i  =  x_{\pi(i)}  
	\end{equation*}
	for all $i \in \Zbin$.
	For convenience we may write\footnote{The meaning of $\pi$ becomes clear from the context.} $y = \pi(x)$ and call $y$ the permutation of $x$.
\end{definition}
According to the introduced notation, the concatenation of two permutations $\pi_1$ and $\pi_2$ is written as $\pi_2 \circ \pi_1 = \pi_2 \pi_1$, to be applied from right to left, i.e., $\pi_1$ first and $\pi_2$ second.  
To apply this to vectors, assume three vectors $x$,  $y = \pi_1(x)$, and $z = \pi_2(y) = \pi_2(\pi_1(x))$;  then $y_i = x_{\pi_1(i)}$ and $z_j = y_{\pi_2(j)}$ from Definition~\ref{def:vector-permutation}, and $z_j =  x_{\pi_1(\pi_2(j))}$ by substituting $i = \pi_1(j)$. 
Next, we define the group of affine transformations over binary vector, and we show how these transformations are related to permutations.
\begin{definition}[Affine transformations]
The General Affine (GA) group $\GA(n)$ is the group of affine transformations of binary vectors $v \in \Fbin^n$,
	\begin{align*}
		T_{(A,b)} : \Fbin^n & \rightarrow \Fbin^n  \\
		v                   & \mapsto     A v + b  ,
	\end{align*}
	with invertible matrices $A \in \Fbin^{n \times n}$ and arbitrary vectors $b \in \Fbin^n$.  
	Each element $T_{(A,b)}$ of this group is uniquely identified by a matrix-vector pair $(A,b)$.
\end{definition}


Every affine transformation gives rise to a permutation, and it does this in natural way for codewords of monomial codes through Definition~\ref{def:bin_exp} as shown as follows. 
\begin{definition}[$\GA$ permutations]
\label{def_GA_perm}
The $\GA$ permutations group is the group of permutations over $\Zbin$ defined by affine transformations as $\pi_{(A,b)}(v) = \invbin{T_{(A,b)}\bin{v} }$; the mapping between $\Zbin$ and $\Fbin^n$ is given as \\
%
\begin{center}
\begin{tikzcd}
	\Zbin           \arrow[d,"\pi_{(A,b)}"] \arrow[r,"\bin{\cdot}"]  &  
	\Fbin^n         \arrow[d,"T_{(A,b)}"]                  \\
	\Zbin           \arrow[r,leftarrow,"\invbin{\cdot}"]             &  
	\Fbin^n
\end{tikzcd}
\ifdefined\COLS
	\hspace{.5cm}
\else
	\hspace{2cm}
\fi
\begin{tikzcd}
	v               \arrow[d,"\pi_{(A,b)}"] \arrow[r,"\bin{\cdot}"]  &  
	\bin{v}         \arrow[d,"T_{(A,b)}"]                  \\
	\pi_{(A,b)}(v)  \arrow[r,leftarrow,"\invbin{\cdot}"]             &  
	A \bin{v} + b
\end{tikzcd}
\end{center}
\end{definition}
It is worth noticing that the vice versa is not true, namely that not every permutation can be expressed as an affine transformation. 
Assuming a boolean function $f \in \Fbin^{[n]}$, 
an affine transformation $T_{(A,b)}$ can be applied to $f$, obtaining boolean function $g=T_{(A,b)}(f)$ defined by 
\begin{equation*}
	g(V) = f( A \cdot V + b ) .
\end{equation*}
It is possible to concatenate affine transformations; given $g=T_{(A_1,b_1)}(f)$ and $h=T_{(A_2,b_2)}(g)$, we have that $h=T_{(A,b)}(f) = T_{(A_2,b_2)}\circ T_{(A_1,b_1)}(f)$. 
By definition, we have that
\ifdefined\COLS
	\begin{align}
	h(V) &= g(A_2 \cdot V +b_2) = f(A_1(A_2 \cdot V + b_2)+b_1) = \\
	&=f(A_1 A_2 \cdot V + A_1b_2 + b_1 ).
	\end{align}
\else
	\begin{equation}
	h(V) = g(A_2 \cdot V +b_2) = f(A_1(A_2 \cdot V + b_2)+b_1) = f(A_1 A_2 \cdot V + A_1b_2 + b_1 ).
	\end{equation}
\fi
As a consequence, $T_{(A,b)}$ is defined by matrix $A=A_1A_2$ and vector $b=A_1b_2+b_1$. 
Note that the order is of the matrices is reversed, as compared to the order of the permutations.

Let us consider now $g=T_{(A,b)}(f)$ and its evaluation $x^{(g)} = \eval(g) = (x^{(g)}_0,\ldots,x^{(g)}_{2^n-1})$; this represents a binary vector that is connected to $x^{(f)}$ by a permutation as follows. 
\begin{lemma}
	\label{lem:GA-perm}
	Assume a $\GA$ transform $T_{(A,b)}$ over $\Fbin^n$ and a boolean functions $f \in \Fbin^{[n]}$, then 
	\begin{equation*}
		x^{\left( T_{(A,b)}(f) \right)} = \pi_{(A,b)} \left( x^{(f)} \right),
	\end{equation*}
	where $\pi_{(A,b)}$ follows Definition~\ref{def_GA_perm}.
\begin{proof}
	Let us consider boolean function $g=T_{(A,b)}(f)$; by Definitions~\ref{def:ev_fun} and~\ref{def:vector-permutation}, we have that
	\ifdefined\COLS
		\begin{align}
		x^{\left( T_{(A,b)}(f) \right)}_i &= x^{(g)}_i = g(\bin{i}) = f \left( A\bin{i}+b \right) = \\
		&= f \left( \bin{\pi_{(A,b)}(i)} \right) = x^{(f)}_{\pi_{(A,b)}(i)}. 
		\end{align}
	\else
		\begin{equation}
		x^{\left( T_{(A,b)}(f) \right)}_i = x^{(g)}_i = g(\bin{i}) = f \left( A\bin{i}+b \right) = f \left( \bin{\pi_{(A,b)}(i)} \right) = x^{(f)}_{\pi_{(A,b)}(i)}. 
		\end{equation}
	\fi
			
\end{proof}	
\end{lemma}
As a consequence, binary vector $x^{(g)}$ is a permutation of $x^{(f)}$; this permutation is the one induced by affine transformation $T_{(A,b)}$, as expressed in this scheme:
\begin{center}
\begin{tikzcd}
	f           \arrow[d,"T_{(A,b)}"] \arrow[r,"eval"]  &  
	x^{(f)}          \arrow[d,"\pi_{(A,b)}"]                  \\
	g           \arrow[r,"eval"]             &  
	x^{(g)}
\end{tikzcd}
\end{center}

For convenient (though imprecise) notation we may simply write $\pi_{A,b}(i) = Ai + b$, presuming the equivalence between an integer $i$ and its binary expansions $\bin{i}$; in practice, we will apply affine transformations or the equivalent permutation on both vectors and boolean functions, the meaning becoming clear from the context. 
An example of this construction can be found in Appendix~\ref{app:eval}. 

Operating over binary vectors, Lemma~\ref{lem:GA-perm} associates a codeword permutation to each $\GA$ transform, and we will refer to these as \emph{$\GA$ permutations}; note, however, that not all permutations can be represented by $\GA$ transforms.  
The $\GA$ permutations group is obviously isomorphic to the group of affine transformations $\GA(n)$.  For convenience we refer to $\GA(n)$ and the group of permutations simply as the group $\GA$.  



The relation between $\GA$ transforms and permutations as given in Lemma~\ref{lem:GA-perm} gives rise to the question how this translates to concatenation. This is answered in the following lemma.

\begin{lemma}[Concatenation of $\GA$ permutations]
	\label{lem:GA-concat}
	Given two $\GA$ permutations $\pi_{(A_1,b_1)}$ and $\pi_{(A_2,b_2)}$, their concatenation 
	\begin{equation*}
		\pi_{(A,b)} = \pi_{(A_2,b_2)} \circ \pi_{(A_1,b_1)}, 
	\end{equation*}
	namely when $\pi_{A_1,b_1}$ is applied first and $\pi_{A_2,b_2}$ second, is the $\GA$ permutation $T_{(A,b)}$ defined by
	\begin{equation*}
		A = A_1 A_2 \; , \qquad b = A_1 b_2 + b_1  .
	\end{equation*}
\end{lemma}
\begin{proof}
	As $\GA$ permutations form a group, their composition is obviously also a $\GA$ transform. 	
	The scheme of the concatenation is:  
\begin{center}
\begin{tikzcd}
	f \arrow[to=Z, bend left=50,"T_{(A,b)}"] \arrow[d,"eval"] \arrow[r,"T_{(A_1,b_1)}"]  &  
	g          \arrow[d,"eval"] \arrow[r,"T_{(A_2,b_2)}"]  &  
	|[alias=Z]| h          \arrow[d,"eval"]                   \\
	x^{(f)}\arrow[to=W, bend right=50,"\pi_{(A,b)}"]  \arrow[r,"\pi_{(A_1,b_1)}"]             &  
	x^{(g)}         \arrow[r,"\pi_{(A_2,b_2)}"]             &  
	|[alias=W]| x^{(h)}
\end{tikzcd}
\end{center}
	If we call $\pi_j=\pi_{(A_j,b_j)}$ and $T_j=T_{(A_j,b_j)}$ for $j=1,2$, by Lemma~\ref{lem:GA-perm} we have that
	\ifdefined\COLS
		\begin{align}
		\pi_{2} \circ \pi_{1}(x^{(f)}_i) &= x^{(f)}_{\pi_{1} ( \pi_2(i))} =	x^{(T_1(f))}_{\pi_2(i)} = x^{(g)}_{\pi_2(i)} = \\
		&= x_i^{(T_2(g))} = x_i^{(T_2 \circ T_1(f))}.
		\end{align}
	\else
		\begin{equation}
		\pi_{2} \circ \pi_{1}(x^{(f)}_i) = x^{(f)}_{\pi_{1} ( \pi_2(i))} =	x^{(T_1(f))}_{\pi_2(i)} = x^{(g)}_{\pi_2(i)} = x_i^{(T_2(g))} = x_i^{(T_2 \circ T_1(f))}.
		\end{equation}
	\fi
	
%
\end{proof}
%
Given the equivalence between affine transformations and affine permutations, in the following we may use $\pi$ to define an affine transformation.

In addition to the group $\GA$ itself, we introduce several sub-groups of $\GA$ and the corresponding permutations. 
These sub-groups will be used in the following sections to prove various results concerning the automorphisms group of polar codes. 
The following sub-groups are defined by transformations $T_{(A,b)}$, where $A$ and $B$ assume a peculiar format: 
\begin{itemize}
	\item the \emph{General Linear group} $\GL$, for which $b=0$;
	\item the \emph{Lower-Triangular Affine group} $\LTA$, for which $A$ is lower-triangular;
	\item the \emph{Block-Lower-Triangular Affine group} $\BLTA(S)$, for which $A$ is block-lower-triangular with profile $S = (s_1,s_2,...,s_l)$ with $s_1+s_2+...+s_l=n$, i.e. a block diagonal matrix having non-zero elements below the diagonal; 
	\item the \emph{Upper-Triangular Linear group} $\UTL$, for which $A$ is upper-triangular; 
	\item the \emph{Permutation Linear group} $\PL$, for which $A$ is a permutation matrix;
	\item the \emph{Translation group} $\T$, for which $A$ is the identity matrix.
\end{itemize}    
Similarly to $\GA$ we may use $\GL$, $\LTA$, $\BLTA$, $\UTL$, and $\PL$ to refer to the corresponding groups of permutations.

\subsection{The affine automorphisms group of polar codes}
\label{sec:ga-automorph-prop}

\begin{definition}[Automorphism]
	A permutation $\pi$ is called an automorphism of code $\code{C}$ if $\pi(\code{C}) = \code{C}$, i.e., if $\pi(x) \in \code{C}$ for all $x \in \code{C}$.  The set of automorphisms of a code $\code{C}$ is denoted by $\Aut(\code{C})$ and forms a group.
\end{definition}

For monomial codes, the permutations from $\GA$ are of particular interest, since they represent linear operations on the variables.  
If a $\GA$ permutation is an automorphism, we call it a $\GA$ (or affine) automorphism.  
We further denote the \emph{group of affine automorphisms} of a code $\code{C}$ by
\begin{equation}
	\mathcal{A}  \triangleq  \GA \cap \Aut(\code{C})  .
	\label{eq:GA-automorphisms}
\end{equation} 

We consider binary $(N,K)$ polar codes $\code{C}$ of length $N=2^n$ and dimension $K$, having information set $\mathcal{I}$ and monomial set $\mathcal{G}$ and following UPO framework. 
In this section, we will prove that the affine automorphisms group of such a polar code is a $\BLTA$ group. 
This property has been proved in \cite{geiselhart2021automorphismPC,li2021complete}; in this paper, we provide an alternative proof based on the introduction of elementary permutations and on the algebraic structure of the $\BLTA$ group. 
 
To begin with, we notice that $\GA$ permutations map monomials of $\mathcal{G}$ to linear combinations of monomials in $\mathcal{M}^{[n]}$. 
A property of $\GA$ automorphisms is their capacity to map the generating monomial set into itself. 
\begin{lemma}
\label{lem:perm_sum}
$\pi \in \mathcal{A}$ if and only if for every $m_t \in \mathcal{G}$, then for every $m_{t_1},\ldots,m_{t_s}$ such that
\begin{equation}
	\pi(m_t) = m_{t_1} + \ldots + m_{t_s}, 
\end{equation}
we have that $m_{t_i} \in \mathcal{G}$.
\begin{proof}
This condition is obviously necessary: if all $m_{t_i} \in \mathcal{G}$, then also $\pi(m_{t}) \in \mathcal{G}$ for every $m_{t} \in \mathcal{G}$. 
Conversely, if $\pi \in \mathcal{A}$, then also $\pi(m_{t}) \in \mathcal{G}$. 
However, in order for the polynomial $m_{t_1} + \ldots + m_{t_s}$ to be included into $\langle \mathcal{G} \rangle$, all its addends must belong to $\mathcal{G}$ because monomials form a base for the code space. 
\end{proof}
\end{lemma}

Equipped with these definitions, we now focus on the characterization of $\mathcal{A}$. 
To begin with, we prove that $\LTA \subseteq \mathcal{A}$ for any polar code, namely that all the $\LTA$ transformations are automorphisms; it is worth noticing that this property holds only for polar codes fulfilling the UPO framework. 
In order to prove it, we denote the row-addition elementary matrix by $\elmat{i}{j}$, namely the $n \times n$ square matrix with ones on the diagonal and an additional one at row $i$ and column $j$. 
This matrix is associated to the elementary linear transformation 
\begin{equation}
	\elmat{i}{j} : \bar{V}_i \rightarrow \bar{V}_i + \bar{V}_j + 1,
\end{equation}
defined by permutation $\elperm{i}{j} = \pi_{(\elmat{i}{j},0)}$. 
This elementary matrix is of particular interest since every $\LTA$ transformation can be decomposed as the product of elementary linear transformations having their one below the diagonal, plus a translation. 

\begin{lemma}
\label{lem:elem_dec}
Every $\pi_{(A,b)} \in \LTA$ can be decomposed as
\begin{equation}
	\pi_{(A,b)} = \elperm{i_q}{j_q} \circ \ldots \circ \elperm{i_1}{j_1} \circ \tau,
\end{equation}
where $\elperm{i}{j}$ is an elementary linear transformation, $\tau \in \T$ and $q$ is the number of nonzero entries of $A$ below the diagonal. 
\begin{proof}
To begin with, we prove that every lower triangular matrix $A$ can be written as a product of $q$ row-addition elementary matrices. 
To do it, we sort the nonzero entries of $A$ below the diagonal from top to bottom and then from left to right: in practice, if $a_{i_l,j_l}$ and $a_{i_{l+1},j_{l+1}}$ are nonzero entries of $A$, then $i_l \leq i_{l+1}$ and, if $i_l = i_{l+1}$, then $j_l < j_{l+1}$; in this way, we have that
\begin{equation}
	A = \prod_{l=1}^{q}\elmat{i_l}{j_l}.  
\end{equation}
In fact, given a matrix $B$, we have that $B \cdot \elmat{i}{j}$ is the matrix produced from $B$ by adding column $i$ to column $j$; we have that $i_l > j_l$ since $A$ is lower triangular, hence realizing the product from left to right results in adding at every step a nonzero entry in position $(i_l,j_l)$, resulting in matrix $A$. 
Next, if we define $\tau = \pi_{I,b}$, the lemma is proved.
\end{proof}
\end{lemma}

When applied to monomials, however, these elementary transformations may lead to polynomials, making it difficult to connect monomial sets. 
In fact, given $m_t \in \mathcal{M}^{[n]}$, then 
\begin{equation}
	\label{eq:el_def}
	\elperm{i}{j}(m_t) = 
	\left\{ 
	\begin{array}{lc}
		m_t & \text{if } i \notin Q \text{ or } i,j \in Q \\
		m_t + m_{t'} + m_{t''} & \text{if } i \in Q \text{ and } j \in Q  
	\end{array} 
	\right.
\end{equation}
where $t'$ is obtained from $t$ by swapping entries $i$ and $j$ of its binary expansion and $t''$ by adding a one in position $i$ of the binary expansion of $t$. 
Lemma~\ref{lem:perm_sum} shows that $\elperm{i}{j} \in \mathcal{A}$ if and only if for every $m_t \in \mathcal{G}$, then $m_{t'}$ and $m_{t''}$ obtained from \eqref{eq:el_def} belong to $\mathcal{G}$. 
In order to prove the first main result of this section, we need to prove that we can focus our analysis only on elementary linear transformations, neglecting the effect of translations. 
The following lemma gives us this possibility. 
\begin{lemma}
If $\code{C}$ follows UPO framework, then $\T \subset \mathcal{A}$.
\begin{proof}
Translation $\tau_i \in \T$ maps $\bar{V}_i$ to $\bar{V}_i+1$. 
Then, it maps each monomial in $\mathcal{G}$ including $\bar{V}_i$ into the sum of the monomial itself and the same monomial without $\bar{V}_i$; since the last monomial is included in $\mathcal{G}$ due to the UPO hypothesis, then $\T \subset \mathcal{A}$. 
\end{proof}
\end{lemma}
Since translations are always automorphisms, we can focus on automorphisms in $\GL$, and more in details we can use elementary linear transformations to prove the following lemma, that is the first main result of this section.
\begin{theorem}
\label{theo_LTA_UPO}
$\LTA \subseteq \mathcal{A}$ if and only if $\code{C}$ follows UPO framework.
\begin{proof}
To begin with, we show that this is a necessary condition. 
The matrix $A$ of a given $\LTA$ transformation can be decomposed as a product of elementary matrices having their extra one below the diagonal. 
for any of these elementary transformations $\elperm{i}{j}$, any $m_t \in \mathcal{G}$ is transformed into a polynomial $m_t + m_{t'} + m_{t''}$ (or remains the same). 
By definition, $m_{t''} \preceq m_t$, and then, by UPO hypothesis, also $m_{t''} \in \mathcal{G}$; moreover, since by construction $i<j$, also $m_{t'} \preceq m_t$, and again $m_{t'} \in \mathcal{G}$ by UPO hypothesis. 
Since all the monomials forming it belong to $\mathcal{G}$, then also $\elperm{i}{j}(m_t)\in \mathcal{G}$ and $\elperm{i}{j} \in \Aut(\code{C})$ when $i<j$; since the product of automorphisms is still an automorphism, we have that $\LTA \subseteq \mathcal{A})$. 
		
To show that it is a sufficient condition, let us proceed by absurd. 
We suppose $\code{C}$ not following the UPO framework, namely there exists some $a\in \mathcal{I}$ for which at least one integer $b$, $a \preceq b$, such that $b \notin \mathcal{I}$. 
It is then possible to create a chain of integers $a = t_0 \preceq t_{1} \preceq \ldots \preceq t_{r-1} \preceq t_r = b$ such that the passage from an integer to the next one is performed by one of the UPO basic rules. 
By absurd hypothesis, there exists an index $s$ such that $t_{s} \in \mathcal{I}$ and $t_{s+1} \notin \mathcal{I}$. 
Then, it is possible to find two integers $i,j$, with $i>j$, such that $\elperm{i}{j}(m_{t_{s}}) = m_{t_{s+1}} + m_{t'}$, either by taking the indices of the swapped entries if $HW(t_{s}) = HW(t_{s+1})$ or setting $i$ as the index of the added one otherwise. 
By construction, $\elperm{i}{j} \notin \mathcal{A}$ and then $\LTA \not\subseteq \mathcal{A}$. 
\end{proof}
\end{theorem}
Theorem~\ref{theo_LTA_UPO} proves that, if the polar code follows the UPO framework, then $\LTA \subseteq \mathcal{A}$. 
Now we expand this result by proving that the affine automorphisms group of a polar code fulfilling UPO has a $\BLTA$ structure. 
To begin with, we prove a lemma regarding elementary linear transformations.
\begin{lemma}
\label{lem:blocky}
If $\code{C}$ follows UPO framework, then  $\elperm{i}{j} \in \mathcal{A}$ implies that also $\elperm{i+1}{j},\elperm{i}{j-1} \in \mathcal{A}$. 
\begin{proof}
According to the hypothesis, for every $m_t \in \mathcal{G}$ and $\elperm{i}{j} \in \mathcal{A}$ such that $\elperm{i}{j}(m_t) = m_t + m_{t'} + m_{t''}$, then $m_{t'},m_{t''} \in \mathcal{G}$. 
Then if $\elperm{i+1}{j}(m_t) = m_t + m_{t_1} + m_{t_2}$, by construction we have that $m_{t_1} \preceq m_{t'}$ and $m_{t_2} \preceq m_{t''}$, and hence by UPO hypothesis $m_{t_1},m_{t_2} \in \mathcal{G}$ and $\elperm{i+1}{j} \in \mathcal{A}$ by Lemma~\ref{lem:perm_sum}. 
Similarly, if $\elperm{i}{j-1}(m_t) = m_t + m_{t_3} + m_{t''}$, by construction we have that $m_{t_3} \preceq m_{t'}$ and hence by UPO hypothesis $m_{t_3} \in \mathcal{G}$ and $\elperm{i}{j-1} \in \mathcal{A}$ by Lemma~\ref{lem:perm_sum}. 
\end{proof}
\end{lemma}
Lemma~\ref{lem:blocky} will be used to prove the second main result of this section. 
\begin{theorem}
\label{theo:UPO_BLTA}
A polar code $\code{C}$ is compliant with the UPO framework if and only if for some profile $s$ we have that $\mathcal{A} = \BLTA(S)$. 
\begin{proof}
The condition is necessary since $\LTA \subseteq \BLTA(S)$ for any profile $S$, and for Theorem~\ref{theo_LTA_UPO} $\LTA \subseteq \mathcal{A}$ implies that $\code{C}$ is compliant with the UPO framework. 
		
To prove that the condition is sufficient, we begin from the observation that Lemma~\ref{lem:blocky} implies that the affine automorphism group of a code $\code{C}$ following UPO framework has an overlapping block triangular ($\OBLT$) structure. 
In practice, this structure is defined by blocks over the diagonal that can overlap. 
In the following, we show that such a structure cannot be a group, and hence a $\BLTA$ is the only "blocky" matrix structure compliant with Lemma~\ref{lem:blocky}. 
\begin{figure}[t!]
	\centering
	\ifdefined\COLS
		\resizebox{0.35\textwidth}{!}{\begin{tikzpicture}
\usetikzlibrary{decorations.pathreplacing}

\draw  (-4,4) rectangle (4,-4);
\draw  (-4,4) edge (4,-4);
\draw  (0,0) rectangle (-4,4);
\draw  (-2,2) rectangle (4,-4);

\draw [decorate,decoration={brace,amplitude=10pt}]
(-4.1,-4) -- (-4.1,4) node [black,midway,xshift=-0.6cm] 
{\large $n$};
\draw [decorate,decoration={brace,amplitude=10pt}]
(-4,4.1) -- (0,4.1) node [black,midway,yshift=0.6cm] 
{\large $s_1$};
\draw [decorate,decoration={brace,amplitude=10pt}]
(.1,2) -- (.1,0) node [black,midway,xshift=0.6cm] 
{\large $m$};
\draw [decorate,decoration={brace,amplitude=10pt}]
(4,-4.1) -- (-2,-4.1) node [black,midway,yshift=-0.6cm] 
{\large $s_2$};

\draw[dashed]  (-3.9,3) edge (3.9,3);
\draw[dashed]  (1.5,3.9) edge (1.5,-3.9);
\draw[fill=white]  (1.5,3) ellipse (.1 and .1);
\node[] at (4.2,3) {$i$};
\node[] at (1.5,4.2) {$j$};

\draw[rounded corners=3pt]   (-2,3.1) rectangle (0,2.9);
\draw[rounded corners=3pt]   (1.4,2) rectangle (1.6,0);

\node[] at (0.5,5) {indices $n-s_2-1,\ldots,s_1-1$};
\draw  (.5,4.9) edge (-.5,3.1);
\draw  (.5,4.9) edge (1.4,1);

\end{tikzpicture}}
	\else
		\resizebox{0.45\textwidth}{!}{\begin{tikzpicture}
\usetikzlibrary{decorations.pathreplacing}

\draw  (-4,4) rectangle (4,-4);
\draw  (-4,4) edge (4,-4);
\draw  (0,0) rectangle (-4,4);
\draw  (-2,2) rectangle (4,-4);

\draw [decorate,decoration={brace,amplitude=10pt}]
(-4.1,-4) -- (-4.1,4) node [black,midway,xshift=-0.6cm] 
{\large $n$};
\draw [decorate,decoration={brace,amplitude=10pt}]
(-4,4.1) -- (0,4.1) node [black,midway,yshift=0.6cm] 
{\large $s_1$};
\draw [decorate,decoration={brace,amplitude=10pt}]
(.1,2) -- (.1,0) node [black,midway,xshift=0.6cm] 
{\large $m$};
\draw [decorate,decoration={brace,amplitude=10pt}]
(4,-4.1) -- (-2,-4.1) node [black,midway,yshift=-0.6cm] 
{\large $s_2$};

\draw[dashed]  (-3.9,3) edge (3.9,3);
\draw[dashed]  (1.5,3.9) edge (1.5,-3.9);
\draw[fill=white]  (1.5,3) ellipse (.1 and .1);
\node[] at (4.2,3) {$i$};
\node[] at (1.5,4.2) {$j$};

\draw[rounded corners=3pt]   (-2,3.1) rectangle (0,2.9);
\draw[rounded corners=3pt]   (1.4,2) rectangle (1.6,0);

\node[] at (0.5,5) {indices $n-s_2-1,\ldots,s_1-1$};
\draw  (.5,4.9) edge (-.5,3.1);
\draw  (.5,4.9) edge (1.4,1);

\end{tikzpicture}}
	\fi
	\caption{Structure of OBLT matrix.}
	\label{fig:OBLT}
\end{figure}
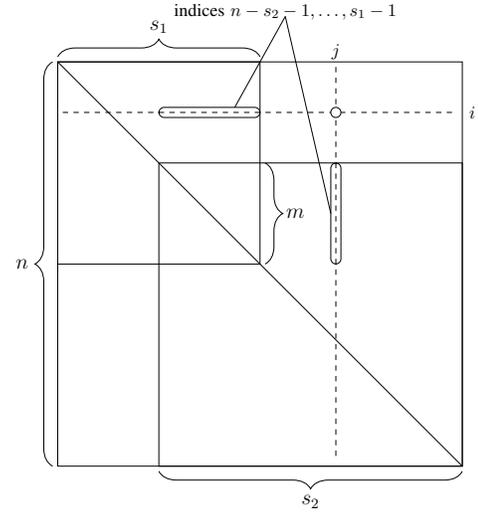 
%
In the following, we restrict our analysis to the case depicted in Figure~\ref{fig:OBLT}, namely an $\OBLT(s_1,s_2)$ structure with $s_1+s_2>n$ and $s_1>s_2$; the results of this case study can be easily extended to more general $\OBLT$ structures. 
In the following, we will show that it is possible to generate a matrix having a non-zero entry in row $i$ and column $j$, for every $0\leq i < n-s_2$ and $n-s_1 \leq j <n$, as the result of the multiplication of two matrices belonging to $\OBLT(s_1,s_2)$. 
This would prove that $\OBLT(s_1,s_2)$ is not a group, and the closure of this set represents the group structure of the automorphisms.
Given $i,j$ defined above, elementary matrices $\elmat{i}{s_1-1}$ and $\elmat{s_1-1}{j}$ belong to $\OBLT(s_1,s_2)$ by construction: the first is included in the upper block, while the other is included in the lower block. 
Matrix $A=\elmat{i}{s_1-1} \cdot \elmat{s_1-1}{j}$ is given by $\elmat{s_1-1}{j}$ but adding row $s_1-1$ to row $i$. 
As a result, $a_{i,j}=1$, and $\OBLT(s_1,s_2)$ is not a group.
\end{proof}
\end{theorem}

\section{Classification of affine automorphisms}
\label{sec:eq_classes}
In this section, we describe a new framework for the analysis of automorphisms for permutation decoding of polar codes. 
We introduce the concept of \emph{permutation decoding equivalence} for automorphisms; this notion permits to cluster code automorphisms into classes of automorphisms always providing the same codeword candidate under permutation decoding, no matter the received signal. 
As a consequence, the selection of automorphisms for AE decoding should avoid automorphisms belonging to the same class. 
These classes are generated as cosets of the absorption group of the polar code, which plays a capital role in the definition of the classes. 
This result is an extension of our preliminary analysis of decoder equivalence published in \cite{AE_class}.

Next, we focus on the analysis of the absorption group under SC decoding, providing an alternative proof of a very recent result presented in \cite{SC_invariant}, namely that the complete SC absorption group has a $\BLTA$ structure. 
Thanks to the knowledge of the structure of the SC absorption group, we calculate the number of the equivalence classes, representing the maximum number of permutations providing possibly different codeword candidates under AE-SC decoding. 
Finally, we propose a practical method to find one automorphism for each equivalence class, in order to avoid redundant automorphisms. 

\subsection{Permutation decoding equivalence}
\label{sec:dec_eq}

We begin by formally defining what is a decoding function of polar codes. 
We consider transmission over a binary symmetric memoryless channel with output $y$ such that
\begin{equation}
	p( y_i | x_i = 0 ) = p( - y_i | x_i = 1 )
	\label{eq:symmetric-channel}
\end{equation} 
for $i=0,1,\ldots, N-1$.  An example of such a channel is the BI-AWGN channel,
\begin{equation*}
	y = \tilde{x} + w ,
\end{equation*}
where $\tilde{x} = \bpsk(x)$ with BPSK mapping $\bpsk(0) = +1$, $\bpsk(1) = -1$, and $w$ is i.i.d. Gaussian noise with distribution $\mathcal{N}(0,\sigma^2)$.
We assume a decoding algorithm for the AWGN channel, operating on the channel outputs\footnote{If the decoder operates on LLRs, we may include the LLR computation into the channel model.} $y$.
\begin{definition}[Decoding function]
\label{def:dec}
The function
\begin{align}
	\dec  :  \realnumber^N  &\rightarrow  \code{C}       \notag\\
	y   	           &\mapsto      x = \dec(y)  
\end{align}
denotes the decoding function for polar codes.
\end{definition}
This definition can be generalized to the notion of decoding function wrapped by an automorphism. 
An \emph{automorphism decoder} is a decoder run on a received signal that is scrambled according to a code automorphism; the result is then scrambled back to retrieve the original codeword estimation.
\begin{definition}[Automorphism decoding function]
\label{def:adec}
For an automorphism $\pi \in \Aut(\code{C})$, the function
\begin{align}
	\adec  :  \realnumber^N \times \Aut(\code{C})  &\rightarrow  \code{C}  \\
		y                                           &\mapsto      x = \adec(y;\pi)  
\end{align}
with
\begin{equation}
	\adec(y;\pi) \triangleq \pi^{-1}( \dec( \pi( y ) ) ) .
\end{equation}
is called the automorphism decoding function of the polar code.
\end{definition}
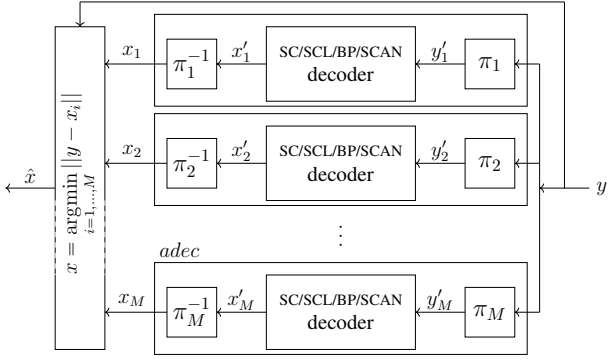
\begin{figure}[t!]
	\centering
	\ifdefined\COLS
		\resizebox{0.45\textwidth}{!}{\begin{tikzpicture}

\draw[ ] (8.75,0) node [font=\large] {$y$};
\draw[<-]  (7.5,0)  -- (8.5,0) ;
\draw  (-1.75,3.75)  -- (8,3.75) ;
\draw  (8,0)  -- (8,3.75) ;
\draw[->]  (-1.75,3.75)  -- (-1.75,3.25) ;
\draw  (7.5,-2.5)  -- (7.5,2.5) ;
\draw[<-]  (7,2.5)  -- (7.5,2.5) ;
\draw[<-]  (7,0.5)  -- (7.5,0.5) ;
\draw[<-]  (7,-2.5)  -- (7.5,-2.5) ;
\draw[<-]  (5,2.5)  -- (6,2.5) ;
\draw[ ] (5.5,2.75) node [font=\large] {$y'_1$};
\draw[<-]  (5,0.5)  -- (6,0.5) ;
\draw[ ] (5.5,0.75) node [font=\large] {$y'_2$};
\draw[<-]  (5,-2.5)  -- (6,-2.5) ;
\draw[ ] (5.5,-2.25) node [font=\large] {$y'_M$};
\draw[<-]  (1,2.5)  -- (2,2.5) ;
\draw[ ] (1.5,2.75) node [font=\large] {$x'_1$};
\draw[<-]  (1,0.5)  -- (2,0.5) ;
\draw[ ] (1.5,0.75) node [font=\large] {$x'_2$};
\draw[<-]  (1,-2.5)  -- (2,-2.5) ;
\draw[ ] (1.5,-2.25) node [font=\large] {$x'_M$};
\draw[<-]  (-1.25,2.5)  -- (0,2.5) ;
\draw[ ] (-0.7,2.75) node [font=\large] {$x_1$};
\draw[<-]  (-1.25,0.5)  -- (0,0.5) ;
\draw[ ] (-0.7,0.75) node [font=\large] {$x_2$};
\draw[<-]  (-1.25,-2.5)  -- (0,-2.5) ;
\draw[ ] (-0.7,-2.25) node [font=\large] {$x_M$};
\draw[<-]  (-3.25,-0)  -- (-2.25,-0) ;
\draw[ ] (-2.75,0.25) node [font=\large] {$\hat{x}$};

\draw[fill=white]  (0,3) rectangle (1,2);
\draw[ ] (0.5,2.5) node [font=\Large] {$\pi_1^{-1}$};
\draw[fill=white]  (0,1) rectangle (1,0);
\draw[ ] (0.5,0.5) node [font=\Large] {$\pi_2^{-1}$};
\draw[fill=white]  (0,-2) rectangle (1,-3);
\draw[ ] (0.5,-2.5) node [font=\Large] {$\pi_M^{-1}$};

\draw[fill=white]  (2,3.25) rectangle (5,1.75);
\draw[ ] (3.5,2.5) node [font=\large] {$\begin{array}{c} \text{\small SC/SCL/BP/SCAN} \\ \text{decoder} \end{array}$};
\draw[fill=white]  (2,1.25) rectangle (5,-0.25);
\draw[ ] (3.5,0.5) node [font=\large] {$\begin{array}{c} \text{\small SC/SCL/BP/SCAN} \\ \text{decoder} \end{array}$};
\path (3.5,1) -- (3.5,-3) node [font=\large, midway, sloped] {$\dots$};
\draw[fill=white]  (2,-1.75) rectangle (5,-3.25);
\draw[ ] (3.5,-2.5) node [font=\large] {$\begin{array}{c} \text{\small SC/SCL/BP/SCAN} \\ \text{decoder} \end{array}$};

\draw[fill=white]  (6,3) rectangle (7,2);
\draw[ ] (6.5,2.5) node [font=\Large] {$\pi_1$};
\draw[fill=white]  (6,1) rectangle (7,0);
\draw[ ] (6.5,0.5) node [font=\Large] {$\pi_2$};
\draw[fill=white]  (6,-2) rectangle (7,-3);
\draw[ ] (6.5,-2.5) node [font=\Large] {$\pi_M$};

%

\draw[fill=white]  (-2.25,3.25) rectangle (-1.25,-3.25);
\draw[dashed, white]  (-2.25,-0.25) edge (-2.25,-1.75);
\draw[dashed, white]  (-1.25,-0.25) edge (-1.25,-1.75);
\draw[ ] (-1.75,0.1) node [rotate=90,font=\large] {$x=\underset{i=1,\ldots, M}{\mathrm{argmin}}\,||y-x_i||$};

\draw  (-0.25,3.5) rectangle (7.25,1.65);
\draw  (-0.25,1.5) rectangle (7.25,-.35);
\draw  (-0.25,-1.5) rectangle (7.25,-3.35);
\draw[ ] (0.25,-1.25) node [font=\large] {$adec$};

\end{tikzpicture}}
	\else
		\resizebox{0.75\textwidth}{!}{\begin{tikzpicture}

\draw[ ] (8.75,0) node [font=\large] {$y$};
\draw[<-]  (7.5,0)  -- (8.5,0) ;
\draw  (-1.75,3.75)  -- (8,3.75) ;
\draw  (8,0)  -- (8,3.75) ;
\draw[->]  (-1.75,3.75)  -- (-1.75,3.25) ;
\draw  (7.5,-2.5)  -- (7.5,2.5) ;
\draw[<-]  (7,2.5)  -- (7.5,2.5) ;
\draw[<-]  (7,0.5)  -- (7.5,0.5) ;
\draw[<-]  (7,-2.5)  -- (7.5,-2.5) ;
\draw[<-]  (5,2.5)  -- (6,2.5) ;
\draw[ ] (5.5,2.75) node [font=\large] {$y'_1$};
\draw[<-]  (5,0.5)  -- (6,0.5) ;
\draw[ ] (5.5,0.75) node [font=\large] {$y'_2$};
\draw[<-]  (5,-2.5)  -- (6,-2.5) ;
\draw[ ] (5.5,-2.25) node [font=\large] {$y'_M$};
\draw[<-]  (1,2.5)  -- (2,2.5) ;
\draw[ ] (1.5,2.75) node [font=\large] {$x'_1$};
\draw[<-]  (1,0.5)  -- (2,0.5) ;
\draw[ ] (1.5,0.75) node [font=\large] {$x'_2$};
\draw[<-]  (1,-2.5)  -- (2,-2.5) ;
\draw[ ] (1.5,-2.25) node [font=\large] {$x'_M$};
\draw[<-]  (-1.25,2.5)  -- (0,2.5) ;
\draw[ ] (-0.7,2.75) node [font=\large] {$x_1$};
\draw[<-]  (-1.25,0.5)  -- (0,0.5) ;
\draw[ ] (-0.7,0.75) node [font=\large] {$x_2$};
\draw[<-]  (-1.25,-2.5)  -- (0,-2.5) ;
\draw[ ] (-0.7,-2.25) node [font=\large] {$x_M$};
\draw[<-]  (-3.25,-0)  -- (-2.25,-0) ;
\draw[ ] (-2.75,0.25) node [font=\large] {$\hat{x}$};

\draw[fill=white]  (0,3) rectangle (1,2);
\draw[ ] (0.5,2.5) node [font=\Large] {$\pi_1^{-1}$};
\draw[fill=white]  (0,1) rectangle (1,0);
\draw[ ] (0.5,0.5) node [font=\Large] {$\pi_2^{-1}$};
\draw[fill=white]  (0,-2) rectangle (1,-3);
\draw[ ] (0.5,-2.5) node [font=\Large] {$\pi_M^{-1}$};

\draw[fill=white]  (2,3.25) rectangle (5,1.75);
\draw[ ] (3.5,2.5) node [font=\large] {$\begin{array}{c} \text{\small SC/SCL/BP/SCAN} \\ \text{decoder} \end{array}$};
\draw[fill=white]  (2,1.25) rectangle (5,-0.25);
\draw[ ] (3.5,0.5) node [font=\large] {$\begin{array}{c} \text{\small SC/SCL/BP/SCAN} \\ \text{decoder} \end{array}$};
\path (3.5,1) -- (3.5,-3) node [font=\large, midway, sloped] {$\dots$};
\draw[fill=white]  (2,-1.75) rectangle (5,-3.25);
\draw[ ] (3.5,-2.5) node [font=\large] {$\begin{array}{c} \text{\small SC/SCL/BP/SCAN} \\ \text{decoder} \end{array}$};

\draw[fill=white]  (6,3) rectangle (7,2);
\draw[ ] (6.5,2.5) node [font=\Large] {$\pi_1$};
\draw[fill=white]  (6,1) rectangle (7,0);
\draw[ ] (6.5,0.5) node [font=\Large] {$\pi_2$};
\draw[fill=white]  (6,-2) rectangle (7,-3);
\draw[ ] (6.5,-2.5) node [font=\Large] {$\pi_M$};

%

\draw[fill=white]  (-2.25,3.25) rectangle (-1.25,-3.25);
\draw[dashed, white]  (-2.25,-0.25) edge (-2.25,-1.75);
\draw[dashed, white]  (-1.25,-0.25) edge (-1.25,-1.75);
\draw[ ] (-1.75,0.1) node [rotate=90,font=\large] {$x=\underset{i=1,\ldots, M}{\mathrm{argmin}}\,||y-x_i||$};

\draw  (-0.25,3.5) rectangle (7.25,1.65);
\draw  (-0.25,1.5) rectangle (7.25,-.35);
\draw  (-0.25,-1.5) rectangle (7.25,-3.35);
\draw[ ] (0.25,-1.25) node [font=\large] {$adec$};

\end{tikzpicture}}
	\fi
	\caption{Structure of the automorphism ensemble (AE) decoder.}
	\label{fig:AE}
\end{figure} 
Note that this function may be seen as a decoding function with parameter $\pi$.  
Moreover, we focus our analysis on affine automorphisms in $\mathcal{A} \subseteq \Aut(\code{C})$. 
An \emph{automorphism ensemble} (AE) decoder, originally proposed in \cite{geiselhart2020automorphism} for Reed-Muller codes, consists of $M$ automorphism decoders running in parallel, as depicted in Figure~\ref{fig:AE}, where the codeword candidate is selected using a least-squares metric.
Starting from these definitions, we analyze the effects of automorphism decoding on polar codes.
\begin{definition}[Decoder equivalence]
\label{def:dec-equiv}
Two automorphisms $\pi_1,\pi_2 \in \mathcal{A}$ are called equivalent with respect to $\dec$, written as $\pi_1 \sim \pi_2$, if
\begin{equation}
	\adec(y;\pi_1) = \adec(y;\pi_2) \quad \text{for all $y \in \realnumber^N$} .
\end{equation}
\end{definition}	
This is an equivalence relation, since it is reflexive, symmetric and transitive. The equivalence classes are defined as
\begin{equation}
	[\pi] \triangleq \{ \pi' \in \mathcal{A} : \pi \sim \pi' \}.
\end{equation}
The equivalence class $[\pid]$ of the trivial permutation $\pid$ is also called the set of \emph{decoder-absorbed automorphisms}.
In fact, for all $\pi \in [\pid]$,
\begin{equation}
	\adec(y;\pi) = \dec(y) \quad \text{for all $y \in \realnumber^N$} ,
\end{equation}
i.e., these permutations are absorbed by the decoder, or in other words, the decoding function is invariant to these permutations. 
This happens also when an absorbed permutation is concatenated to other permutations. 
In fact, given two automorphisms $\pi,\sigma \in \mathcal{A}$, then
\begin{equation}
	\adec(y;\pi \circ \sigma) = \sigma^{-1}( \adec( \sigma(y) ; \pi ))  ,
	\label{eq:adec-pi-sigma}
\end{equation}
and if $\pi \in [\pid]$, we obtain
\ifdefined\COLS
	\begin{align}
	\adec(y;\pi \circ \sigma) &= \sigma^{-1}( \adec( \sigma(y) ; \pi )) = \\
	&= \sigma^{-1}( \dec( \sigma(y) )) = \adec(y;\sigma)  ,
	\label{eq:adec-pid-sigma}
	\end{align}
\else
	\begin{equation}
	\adec(y;\pi \circ \sigma) = \sigma^{-1}( \adec( \sigma(y) ; \pi )) = \sigma^{-1}( \dec( \sigma(y) )) = \adec(y;\sigma)  ,
	\label{eq:adec-pid-sigma}
	\end{equation}
\fi

i.e., the component $\pi$ of the composition $\pi \circ \sigma$ is absorbed. In the following, we generalize these properties by showing that $[\pid]$ forms a sub-group of $\mathcal{A}$.
To begin with, we prove that the inverse of an absorbed automorphism is also absorbed. 
\begin{lemma}
If $\pi \in [\pid]$, then $\pi^{-1} \in [\pid]$.
\begin{proof}
From the definition of the equivalence class, we have that
\begin{equation}
    \dec(y) = \pi^{-1}( \dec( \pi( y ) ) )  
    \,\,\Leftrightarrow\,\,
   \pi(\dec(y)) = \dec( \pi( y ) ) 
\end{equation}
(with the latter denoting an equivariance); and so
\ifdefined\COLS
	\begin{align}
	\adec(y;\pi^{-1}) &= \pi( \dec( \pi^{-1}( y ) ) ) = \\
	&= \dec( \pi( \pi^{-1} (y) ) ) = \dec(y).
	\end{align}
\else
	\begin{equation}
    \adec(y;\pi^{-1}) = \pi( \dec( \pi^{-1}( y ) ) ) = \dec( \pi( \pi^{-1} (y) ) ) = \dec(y).
	\end{equation}
\fi
\end{proof}
\end{lemma}
Now we can prove that $[\pid] \leq \mathcal{A}$;
\begin{lemma}
	The equivalence class $[\pid]$ (set of decoder-absorbed automorphisms) is a subgroup of $\mathcal{A}$, i.e., $[\pid] \leq \mathcal{A}$.
\begin{proof}
We prove this lemma using the subgroup test, stating that $[\pid] \leq \mathcal{A}$ if and only if $\forall \pi,\sigma \in [\pid]$ then $\pi^{-1}\sigma \in [\pid]$:
\begin{align}
    \adec(y;\pi^{-1}\sigma) &= (\pi^{-1}\sigma)^{-1}( \dec( (\pi^{-1}\sigma)( y ) ) ) =\\
    &= \sigma^{-1}\left(\pi\left( \dec\left( \pi^{-1}\left(\sigma\left( y \right)\right)\right)\right)\right) =\\
    &= \sigma^{-1}\left( \dec\left(\sigma\left( y \right)\right)\right) =\\
    &= \dec\left( y \right).
\end{align}
\end{proof}
\end{lemma}
Now, we use cosets of $[\pid]$ to classify automorphisms into equivalence classes ($\EC$s) containing permutations providing the same results under AE decoding. 
\begin{lemma} \label{lem:ec_cosets}
Equivalence classes $[\sigma]$ defined by decoder equivalence correspond to right cosets of $[\pid]$. 
\begin{proof}
Right cosets of $[\pid]$ are defined as
\begin{equation}
	[\pid]\sigma \triangleq \{ \pi \circ \sigma : \pi \in [\pid] \}.
\end{equation}
Hence, permutations $\sigma_1,\sigma_2 \in [\pid]\sigma$ if and only if there exist two permutations $\pi_1,\pi_2 \in [\pid]$ such that $\sigma_1 = \pi_1 \circ \sigma$ and $\sigma_2 = \pi_2 \circ \sigma$, where the second implies $\sigma = \pi_2^{-1} \circ \sigma_2$. Then the proof is concluded by
\ifdefined\COLS
	\begin{align}
	\adec(y;\sigma_1) &= \adec(y;\pi_1\sigma) = \\
	&= \adec(y;\pi_1\pi_2^{-1}\sigma_2) = \adec(y;\sigma_2).
	\end{align}
\else
	\begin{equation}
	\adec(y;\sigma_1) = \adec(y;\pi_1\sigma) = \adec(y;\pi_1\pi_2^{-1}\sigma_2) = \adec(y;\sigma_2).
	\end{equation}
\fi
\end{proof}
\end{lemma}
According to our notation, two automorphisms in the same $\EC$ always provide the same candidate under $\adec$ decoding.
The number of non-redundant automorphisms for AE-dec, namely the maximum number of different candidates listed by an AE-dec decoder, is then given by the number of equivalence classes of our relation.
\begin{lemma}
\label{lem:number_EC}
There are $E=\frac{|\mathcal{A}|}{|[\pid]|}$ equivalence classes $[\pi]$, $\pi \in \mathcal{A}$, all having the same size.
\end{lemma}
\begin{proof}
Follows from Lemma~\ref{lem:ec_cosets} and Lagrange’s Theorem. 
\end{proof}
The number $E$ of equivalence classes provides an upper bound on the number of different candidates of an $\adec$ decoder. 
In fact, all the permutations included in an equivalence class give the same result under $\adec$ decoding. 
Then, when selecting permutations for the $\adec$ decoder, it is fundamental to select permutations of different equivalence classes. 
Moreover, it is useless to select more permutations than the number of equivalence classes. 
This observation permits to control the list size of an $\adec$ decoder. 
However, the proposed relation does not assure that, for some value of $y$, two $\EC$s do not produce
the same candidate; in practice, our relation permits to calculate the maximum number of different results under adec, providing an upper bound of parameter $M$ of an AE-dec decoder.

Equipped with the definition of equivalence classes under a given decoder, it is worth analyzing the \emph{redundancy} of an automorphism set $\mathcal{L}$ used in an AE decoder.
Let us denote by $M=|\mathcal{L}|$ the number of automorphisms used in the AE decoder; if the elements of $\mathcal{L}$ are randomly drawn from $\mathcal{A} \setminus [\pid]$, we call $P_{\geq 1}(M)$ the probability of having redundant automorphisms in the set, i.e. that at least two automorphisms belong to the same equivalence class and hence provably provide always the same candidates. 
The probability of having non-redundant sets is connected to the birthday problem; if we denote by $E = |\EC|$ the number of equivalence classes under the chosen decoder, this probability is given by
\begin{equation}
P_{\geq 1}(M) = 1-P_0(M) = 1-\prod_{i=0}^{M-1}\frac{E-i}{E},
\end{equation}
where $P_m(M)$ represents the probability of having exactly $m$ redundant automorphisms.
The probability $P_m(M)$ is more difficult to compute, since the $m$ automorphism may belong to the same equivalence classes or from different equivalence classes. 
The second possibility being more likely to happen, we provide a lower bound of $P_m(M)$ as the probability of having $m$ EC with exactly $2$ representatives as
\begin{equation}\label{eq:lowerbound}
P_m(M)\geq \binom{E}{m}\prod_{k=0}^{m-1}\binom{M-2k}{2}\prod_{k=0}^{M-2m-1}\frac{E-i-m}{E}.
\end{equation}
The effect of the redundancy of $\mathcal{L}$ will be shown in Section~\ref{sec:nem_res}. 
In the next section, we will analyze $\EC$s under SC decoding.

\subsection{Absorption group of Successive Cancellation decoders}
\label{sec:SC_abs}

We assume the standard algorithm for SC decoding for the AWGN channel, operating with the min-approximation of the boxplus operation, with the kernel decoding equations
\begin{align*}
	l^-  &=  \sgn(l_0) \cdot \sgn(l_1) \cdot \min\{|l_0|,|l_1|\}  , \\
	l^+  &=  l_0 + l_1 . 
\end{align*}
This decoder is independent\footnote{The decoder with the exact boxplus operation has similar properties, but depends on the SNR.} of the SNR and can directly operate on the channel outputs $y$ rather than (properly scaled) channel LLRs. 
Here we propose again Definitions~\ref{def:dec},\ref{def:adec} for SC decoding. 
\begin{definition}[SC decoding function]
The function
\begin{align}
	\scdec  :  \realnumber^N  &\rightarrow  \code{C}       \notag\\
	y   	           &\mapsto      x = \scdec(y)  
\end{align}
denotes the decoding function implemented by the SC algorithm (with min-approximation).
\end{definition}
\begin{definition}[Automorphism SC decoding function]
	For an automorphism $\pi \in \mathcal{A}$, the function
	\begin{align}
		\ascdec  :  \realnumber^N \times \mathcal{A}  &\rightarrow  \code{C}  \\
		y                                             &\mapsto      x = \ascdec(y;\pi)  
	\end{align}
	with
	\begin{equation}
		\ascdec(y;\pi) \triangleq \pi^{-1}( \scdec( \pi( y ) ) ) .
	\end{equation}
	is called the automorphism SC decoding function.
\end{definition}
Again, $\ascdec$ function may be seen as a decoding function with parameter $\pi$.  
The notion of SC-absorbed automorphisms group $[\pid]$ has been introduced in \cite{geiselhart2020automorphism} as the set of permutations that are SC decoding invariant. 
In the same paper, authors proved that $\LTA$ automorphisms are absorbed under SC decoding; this result is reported in the following lemma:
\begin{lemma}[SC-absorption of $\LTA$]
\label{lem:LTA_SC}
The group $\LTA$ is SC-absorbed, $\LTA \le [\pid]$.
\begin{proof}
The proof is provided in \cite{geiselhart2020automorphism}.  
%
%
\end{proof}
\end{lemma}
In \cite{geiselhart2020automorphism} it was also conjectured that $[\pid] = \LTA$, however this statement is not true. 
In fact, the full absorption group of a polar code may be larger than $\LTA$:
\begin{lemma}
\label{lem:BLTA_2}
If $\BLTA(S)$ is the affine permutation group of a polar code with $s_1>1$, then $\BLTA(2,1,\dots,1) < [\pid]$.
\begin{proof}
Given $\pi_{(A,b)} \in \mathcal{A}$, to set $a_{0,1} = 1$ in matrix $A$ corresponds to map variable $X_0$ to variable $a_{0,0}+X_1$; in practice it represents a permutation that identically scrambles 4-uples of the codeword. 
To be more precise, the entries $a_{0,0},a_{0,1}$ of $A$ represent a permutation of the first 4 entries of the codeword, which is repeated identically for every subsequent block of 4 entries of the vector. 
Under SC decoding, this represents a permutation of the entries of the last $4\times 4$ decoding block; in practice, the difference between the $\ascdec$ decoding of two codewords permuted according two  permutations differing only for entries $a_{0,0},a_{0,1}$ is that entries of the leftmost $4\times 4$ decoding block are permuted. 
Since the polar code follows UPO as stated by Theorem~\ref{theo:UPO_BLTA}, each block of 4 entries of the input vector can be described as one of the following sequences of frozen (\textit{F}) and information (\textit{I}) bits, listed in increasing rate order with the notation introduced for fast-SC decoders \cite{Fast-SC_SPCREP}; 
\begin{itemize}
	\item \textit{[FFFF]}: this represents a rate-zero node, and returns a string of four zeroes no matter the input LLRs; this is independent of the permutation. 
		\item \textit{[FFFI]}: this represents a repetition node, and returns a string of four identical bits given by the sign of the sum of the LLRs; this is independent of the permutation. 
		\item \textit{[FFII]}: this case is not possible if $s_1>1$.
		\item \textit{[FIII]}: this represents a single parity check node, and returns the bit representing the sign of each LLR while the smallest LLR may be flipped if the resulting vector has even Hamming weight; permuting them back gives the same result for the two decoders. 
		\item \textit{[IIII]}: this represents a rate-one node, and returns the bit representing the sign of each LLR; permuting them back gives the same result for the two decoders. 
	\end{itemize}
	As a consequence, changing entries $a_{0,0},a_{0,1}$ of the permutation matrix (while keeping it invertible) does not change the result of the $\ascdec$ decoder.
	\end{proof}
\end{lemma}
Recently, authors in \cite{SC_invariant} extended this result to other $\BLTA$ structures, proving that, under SC decoding, $[\pid]$ is a $\BLTA$ space, and providing guidance to find the profile of such a group. 
Here we reconsider this result, providing a non-constructive proof based on algebraic reasoning that does not consider the frozen set of the code.
\begin{theorem}
\label{theo:BLTA_pid}
If $\mathcal{A} = \BLTA(S)$, then $[\pid] = \BLTA(S_{\pid})$, with 
\ifdefined\COLS
	$S_{\pid} = (s_{1,1},s_{1,2},\ldots,s_{1,j_1},s_{2,1},\ldots,s_{2,j_2},\ldots,s_{s_l,1},\ldots,$ $s_{s_l,j_l})$ 
\else
	$S_{\pid} = (s_{1,1},s_{1,2},\ldots,s_{1,j_1},s_{2,1},\ldots,s_{2,j_2},\ldots,s_{s_l,1},\ldots,s_{s_l,j_l})$ 
\fi
where $s_{i,1}+\ldots+s_{i,j_i}=s_i$ for every $1\leq i \leq l$. 
\begin{proof}
The proof follows from Lemma~\ref{lem:BLTA_2} and from the same line of reasoning of the proof of sufficient condition of Theorem~\ref{theo:UPO_BLTA}.
\end{proof}
\end{theorem}
In practice, to construct the profile $S_{\pid}$, the block of size $s_i$ of $S$ is divided into sub-blocks of size $s_{i,1},\ldots,s_{i,j_i}$.
Given that both $\mathcal{A}$ and $[\pid]$ have a $\BLTA$ structure under SC decoding, in order to count the number of $\EC$s under this decoder we need to calculate the size of such a group. 
\begin{lemma}\label{lem:sizeBLTA}
The size of $\BLTA(S)$, with $S=(s_1,\dots,s_t)$ and $\sum_{i=1}^t s_i = n$, is: 
\begin{align}\label{eq:sizeBLTA}
	|\BLTA(S)| &= 2^{\frac{n(n+1)}{2}}\cdot \prod_{i=1}^t \left( \prod_{j=2}^{s_i} \left( 2^j-1 \right) \right).
\end{align}
\begin{proof}
It is well known that $|\GL(m)| = \prod_{i=0}^{m-1} \left( 2^m-2^i \right)$. 
From this, we have that
\begin{align*}
	|\GL(m)| & = \prod_{i=0}^{m-1} \left( 2^m-2^i \right) = \\
	& = \prod_{i=0}^{m-1} 2^i \left( 2^{m-i}-1 \right) = \\
	& = \left( \prod_{j=0}^{m-1} 2^j \right) \cdot \left( \prod_{i=0}^{m-1} \left( 2^{m-i}-1 \right) \right) = \\
	& = 2^{\sum_{j=0}^{m-1} j} \cdot \prod_{i=1}^{m} \left( 2^i-1  \right) = \\
	& = 2^{\frac{m(m-1)}{2}} \cdot \prod_{i=2}^{m} \left( 2^i-1  \right).
\end{align*}
It is worth noting that we rewrote the size of $\GL(m)$ as the product of the number of lower-triangular matrices and the product of the first $n$ powers of two, diminished by one. 
This property can be used to simplify the calculation of $\BLTA(S)$.  
In fact, each block of the $\BLTA$ structure forms an independent $\GA(s_i)$ space, having size $|\GA(s_i)|$. 
All the entries above the block diagonal are set to zero, so they are not taken into account in the size calculation, while the entries below the block diagonal are free, and can take any binary value. 
Then, the size of $\BLTA(S)$ can be calculated as $|\LTA(n)| = 2^{\frac{n(n+1)}{2}}$ multiplied by the product of the first $s_i$ powers of two, diminished by one, for each size block $s_i$, which concludes the proof. 
\end{proof}
\end{lemma}
Now we can prove the last result of this section, namely the number of $\EC$s under SC decoding. 
\begin{lemma}\label{lem:size_EC_SC}
A polar code of length $N=2^n$ with $\mathcal{A} = \BLTA(S)$, $S=(s_1,\dots,s_t)$ and $[\pid] = \BLTA(S_{\pid})$ defined in Theorem~\ref{theo:BLTA_pid} has 
\begin{equation}\label{eq:nb_ec_sc}
E=\frac{|\BLTA(S)|}{|\BLTA(S_{\pid})|} = \frac{\prod_{i=1}^t \left( \prod_{j=2}^{s_i} \left( 2^j-1 \right) \right)}{\prod_{i=1}^t \prod_{l=1}^{j_i} \left( \prod_{r=2}^{s_{i,l}} \left( 2^r-1 \right) \right)}.
\end{equation}
equivalence classes under SC decoding.
	\begin{proof}
	This follows from the application of Lemma~\ref{lem:number_EC} and Lemma~\ref{lem:sizeBLTA}.
	\end{proof}
\end{lemma}
Lemma~\ref{lem:size_EC_SC} also provides an upper bound on the number of $\EC$s for a polar code under SC decoding; in fact, since $\LTA < [\pid]$, we have that
\begin{equation}
E=\frac{|\BLTA(S)|}{|[\pid]|} \leq \frac{|\BLTA(S)|}{|\LTA|} = \prod_{i=1}^t \left( \prod_{j=2}^{s_i} \left( 2^j-1 \right) \right).
\end{equation}
In the scenario of Lemma~\ref{lem:BLTA_2}, the upper bound provided by Lemma~\ref{lem:size_EC_SC} on the number of $\EC$s for a polar code under SC decoding can be rewritten as
\begin{equation}
E=\frac{|\BLTA(S)|}{|\BLTA(2,1,\ldots,1)|} = \frac{1}{3}\prod_{i=1}^t \left( \prod_{j=2}^{s_i} \left( 2^j-1 \right) \right).
\end{equation}

\subsection{Decomposition of equivalence classes under SC decoding}
\label{sec:decomposition}

Lemma~\ref{lem:number_EC} provides the number of $\EC$s of the $\ascdec$ decoding relation, however without providing a method to calculate them. 
In this section we show how, thanks to the structure of the automorphism group of polar codes, $\EC$ representatives can be decomposed in blocks, which can be calculated separately. 
This decomposition greatly simplifies the task of listing all the $\EC$s of the relation for practical application of our proposal. 

To begin with, we define the block diagonal matrix $B = diag(B_1,\ldots,B_l)$ where $B_i$ is an invertible square matrix in $\GL(s_i)$. 
By definition, $B$ is the transformation matrix of an automorphism of the code; in the following, we will prove that every $\EC$ contains at least an automorphism defined by such a matrix. 
In particular, we call the equivalence sub-class $\EC_i = \GL(s_i) / \BLTL(s_{i,1},\ldots,s_{i,j_i})$, namely the set of the cosets of $\BLTL(s_{i,1},\ldots,s_{i,j_i})$ in $\GL(s_i)$. 
By definition, we have that $\prod_{i=1}^{l} |\EC_i| = |\EC|$ under SC decoding; in the following, we prove that a representative of any $\EC$ can be expressed as the juxtaposition of representatives of equivalence sub-classes. 
\begin{lemma}
\label{lem:EC_dec}
For every $B_i,D_i \in \GL(s_i)$ such that $B=diag(B_1,\ldots,B_l)$ and $D=diag(D_1,\ldots,D_l)$ for all $i=1,\ldots,l$, then $[\pi_{(B,0)}]=[\pi_{(D,0)}]$ if and only if $B_i$ and $D_i$ belong to the same coset of $\GL(s_i) / \BLTA(s_{i,1},\ldots,s_{i,j_i})$.
\begin{proof}
First, assume that $B_i$ and $D_i$ belong to the same equivalence sub-class for all $i=1,\ldots,l$; since each block is independent, 
\begin{equation}
BD^{-1} = diag(B_1 D_1^{-1},\ldots,B_l D_l^{-1}) \in \BLTA(S_{\pid})
\end{equation}
since for every block $B_i D_i^{-1} \in \BLTL(s_{i,1},\ldots,s_{i,j_i})$. 
Conversely, assume $[\pi_{(B,0)}]=[\pi_{(D,0)}]$. 
This means that $BD^{-1} \in \BLTA(S_{\pid})$, and given that each block is independent, $B_i D_i^{-1} \in \BLTL(s_{i,1},\ldots,s_{i,j_i})$ for all $i=1,\ldots,l$. 
\end{proof}
\end{lemma}
Now we can prove that every $\EC$ can be represented by an affine automorphism whose matrix can be expressed as a block diagonal matrix. 
\begin{theorem}
\label{lem:EC_diag}
Every $\EC$ includes at least an affine automorphism whose matrix is block diagonal.
\begin{proof}
Lemma~\ref{lem:EC_dec} proves the connection between equivalence sub-classes and $\EC$s; in particular, we have seen that different compositions of sub-classes lead to different $\EC$s. 
The proof can be conluded by proving that the number of ways to compose block diagonal matrices using representatives of different sub-classes is equal to the number of $\EC$s; in fact,
\begin{align}
\prod_{i=1}^{l} |\EC_i| & = \prod_{i=1}^{l} |\GL(s_i) / \BLTA(s_{i,1},\ldots,s_{i,j_i})| =\\
 & = \prod_{i=1}^{l} \frac{|\GL(s_i)|}{|\BLTA(s_{i,1},\ldots,s_{i,j_i})|} =\\
 & = \prod_{i=1}^{l} \frac{2^{\frac{s_i(s_i-1)}{2}}\prod_{j=2}^{s_i}(2_j-1) }{2^{\frac{s_i(s_i-1)}{2}} \prod_{i=1}^{t} (\prod_{j=2}^{s_{i,j_t}} (2^j-1) ) } =\\
 & = \frac{\prod_{i=1}^t \left( \prod_{j=2}^{s_i} \left( 2^j-1 \right) \right)}{\prod_{i=1}^t \prod_{l=1}^{j_i} \left( \prod_{r=2}^{s_{i,l}} \left( 2^r-1 \right) \right)}  =\\
 & = |\EC|.
\end{align}
\end{proof}
\end{theorem}
It is worth noticing that this property holds when $\LTA \leq [\pid]$; if the absorption group is smaller, the property may not be true. 
This property can be used to efficiently list all the equivalence classes under SC decoding. 
In fact, each block of the $\BLTA$ automorphism group can be searched independently, and the resulting sub-classes can be merged to find all the equivalence classes. 

Focusing on a sub-class $\EC_i$, we propose to use PUL matrix decomposition to further simplify the representatives search. 
\begin{definition}[PUL decomposition]
	The matrices $P \in \PL$, $U \in \UTL$ and $L \in \LTL$ form the PUL decomposition of invertible matrix $A$ if $A = P U L$. 
\end{definition}
The PUL decomposition exists for all $A \in \GL$ and is not necessarily unique; this is due to the non-uniqueness of the LU decomposition, of which the PUL decomposition is merely a variation. 
We extend the notion of PUL decomposition to automorphisms $\pi_{(A,b)} \in \mathcal{A}$, corresponding to the concatenation of the permutations as
\begin{equation}
	\pi_{(A,b)} = \pi_{(L,b)} \circ \pi_{(U,0)} \circ \pi_{(P,0)}  ,
\end{equation}
applied from right to left. 
By Lemma~\ref{lem:LTA_SC}, we have that under SC decoding $\pi_{L,b} \in [\pid]$. 
In the following, we show that representatives of $\EC$s can be decomposed as the product of two automorphisms, one belonging to $\UTL$ and the other belonging to $\PL$. 
If we call $\mathcal{A}_{\mathcal{P}}$ and $\mathcal{A}_{\mathcal{U}}$ the subgroups of $\mathcal{A}$ containing only $\PL$ and $\UTL$ automorphisms respectively, we can always find an $\EC$ representative composing elements from these two sets. 
\begin{lemma}\label{lem:PU_BLT}
Each $\EC$ contains at least an automorphism $P \cdot U$ with $P\in\mathcal{A}_{\mathcal{P}}$ and $U\in\mathcal{A}_{\mathcal{U}}$.
\begin{proof}
By Theorem~\ref{lem:EC_diag} we know that every $\EC$ includes an affine automorphism $\pi_{(D,b)}$ such that $D=diag(D_1,\ldots,D_l)$. 
Every sub-matrix $D_i$ has PUL decomposition $D_i=P_iU_iL_i$, and we define $P=diag(P_1,\ldots,P_l)$, $U=diag(U_1,\ldots,U_l)$ and $L=diag(L_1,\ldots,L_l)$, with $P\in\mathcal{A}_{\mathcal{P}}$ and $U\in\mathcal{A}_{\mathcal{U}}$. 
then, we have that
\begin{align}
\pi_{(D,b)} &= \pi_{(diag(D_1,\ldots,D_l),b)} =\\
            &= \pi_{(diag(P_1U_1L_1,\ldots,P_lU_lL_l),b)} =\\
            &= \pi_{(PUL,b)} =\\
            &= \pi_{(L,b)} \circ \pi_{(U,0)} \circ \pi_{(P,0)} ,
\end{align}
and since $\pi_{(L,b)} \in [\pid]$ we have that $[\pi_{(D,b)}] = [\pi_{(U,0)} \circ \pi_{(P,0)}]$. 
\end{proof}
\end{lemma}
The proof of the previous lemma is constructive, in the sense that it proposes a method to list all $\EC$s by mixing $\UTL$ and $\PL$ automorphisms that are block diagonal. 
In fact, in order to generate an $\EC$ candidate matrix $A$, it is sufficient to randomly draw $l$ $\UTL$ matrices $U_1,\ldots,U_l$ and $l$ $\PL$ matrices $P_1,\ldots,P_l$, of size $s_1,\ldots,s_l$ respectively, create the block diagonal matrices $U=diag(U_1,\ldots,U_l)$ and $P=diag(P_1,\ldots,P_l)$ and generate $A=P \cdot U$. 
By listing all the possible matrices $P$ and $U$, Lemma~\ref{lem:EC_dec} ensures that all the $\EC$s will be found. 

However, since the PUL decomposition of a matrix is not unique, this method is redundant, in the sense that it may generate multiple representatives of the same $\EC$. 
A further check is required to assure that the $\EC$ representative does not belong to an already calculated $\EC$. 
This check is done by multiplying the calculated $P \cdot U$ matrix and the $P \cdot U$ matrices of the previously calculated $\EC$s as stated in Lemma~\ref{lem:checkproductA}. 
\begin{lemma}\label{lem:checkproductA}
Given $\pi_1,\pi_2 \in \mathcal{A}$ having transformation matrices $A_1,A_2$, then $\pi_1\in[\pi_2]$ if and only if $A_1 \cdot A_2^{-1} \in\BLTA(S_{\pid})$.  
\begin{proof}
This follows directly from the definitions of $\EC$ and $[\pid]$. 
\end{proof}
\end{lemma}

We can now propose a practical method to create the automorphisms list $\mathcal{L}$, including an automorphism for every $\EC$. 
The idea is to start by dividing the affine transformation matrix in $l$ blocks of size $s_1,\ldots,s_l$. 
For each block, all the matrices belonging to $\mathcal{U}(s_i)$ and $\mathcal{P}(s_i)$, namely the group of upper triangular and permutation matrices of size $s_i$, are calculated. 
Next, all the block upper diagonal matrices $U\in\mathcal{A}_{\mathcal{U}}$ having matrices in $\mathcal{U}(s_i)$ on the diagonal are calculated, along with all the block permutation matrices $P\in\mathcal{A}_{\mathcal{P}}$ having matrices in $\mathcal{P}(s_i)$ on the diagonal; each block diagonal matrix is considered ad the affine transformation matrix of a representative of an $\EC$ and included in $\mathcal{L}$, after checking that the new automorphisms is not includes in any coset already included in the list as stated in Lemma~\ref{lem:PU_BLT}.  
Next, matrices in the form $P \cdot U$ are included in the list if they pass the check of Lemma~\ref{lem:PU_BLT}. 
When $|\EC|$ representatives are found, the searching process stops. 

In order to evaluate the complexity of the proposed method, we need to evaluate the number of matrices that are generated during the process. 
By construction, we have that 
\begin{equation}
|\mathcal{U}(s_i)|=2^{\frac{s_i(s_i-1)}{2}}  \quad , \quad  |\mathcal{P}(s_i)|=s_i!~.
\end{equation}
For a block structure $S=(s_1,\dots,s_t)$, the number of $\UTL$ and $\PL$ automorphisms are:
\begin{equation}\label{eq:aut_U_P}
	|\mathcal{A}_{\mathcal{U}}| = \prod_{i=1}^{t}2^{\frac{s_i(s_i-1)}{2}} \quad , \quad |\mathcal{A}_{\mathcal{P}}| = \prod_{i=1}^{t}s_i!~.
\end{equation}
Hence, the number of generated matrix is given by 
\begin{equation}
|\mathcal{A}_{\mathcal{U}}|+|\mathcal{A}_{\mathcal{P}}|+|\mathcal{A}_{\mathcal{U}}|\cdot|\mathcal{A}_{\mathcal{P}}|.
\end{equation}
It is worth noticing that if $[\pid]=\LTA$ then every $\EC$ include only one $\UTL$ or $\PL$ matrix, such that $\mathcal{A}_{\mathcal{U}} \cup \mathcal{A}_{\mathcal{P}} \subset \mathcal{L}$. 
An example of this construction can be found in Appendix~\ref{app:rep_list}.

%
%
%
%

\section{Numerical results}
\label{sec:nem_res}

%
%
%


In this section, we present a numerical analysis of error correction performance of AE decoders. 
Simulation results are obtained under BPSK modulation over the AWGN channel.
We analyze polar codes having meaningful affine automorphism groups, decoded under AE-$M$-$\dec$ where $M$ represents the number of parallel $\adec$ decoders.
For each code, we provide the minimum information set $\mathcal{I}_{min}$, the affine automorphism group $\mathcal{A}$ and the SC absorption group $[\pid]$. 
Automorphism set $\mathcal{L}$ are generated according to the method described in Lemma~\ref{lem:checkproductA}; the results obtained are compared to a random selection of automorphisms. 
We compare error-correction performance of proposed polar codes to 5G polar codes under CRC-aided SCL \cite{polar_5G}; ML bounds depicted in the figures are retrieved with the truncated union bound, computed with the minimum distance $d_{min}$ and the number of minimum distance codewords of the code \cite{BardetPolyPC}.

\begin{figure}[t]
  \centering
  \includegraphics[width=0.95\columnwidth]{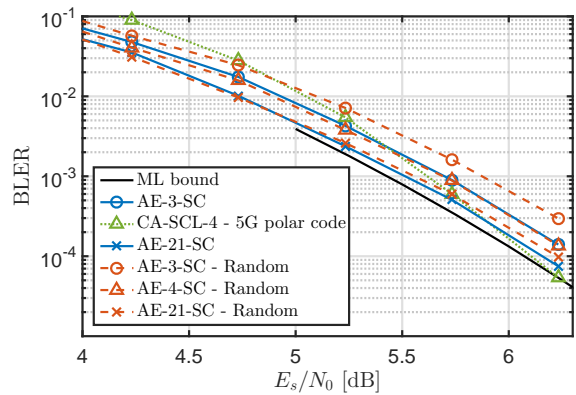}
\caption{Error-correction performance of $(128,85)$ polar codes under AE-SC decoding.}
\label{fig:128_85}
\end{figure}

Figure~\ref{fig:128_85} shows the performance of AE-SC decoding for the $(128,85)$ code defined by $\mathcal{I}_{min}=\{23,25\}$ and having $\mathcal{A}=\BLTA(3,1,3)$.
Following the procedure of Lemma~\ref{lem:BLTA_2}, it is possible to prove that this code has sequences of length $8$ of frozen and information bits that are invariant under SC decoding, thus $[\pid]=\BLTA(3,1,1,1,1)$ \cite{SC_invariant}. 
This code has $E=21$ equivalence classes, so a set $\mathcal{L}$ composed of $21$ $\EC$ representatives provides a bound on the decoding performance of AE-$\scdec$ decoding. 
AE with a set of $M=3$ randomly drawn automorphisms suffers from a significant loss with respect to AE having a set composed of $M=3$ $\EC$ representatives; this loss is eliminated by drawing an additional automorphism, setting $M=4$.
In this case, a set of $21$ random automorphisms is essentially matching the AE bound; the impact of new non-redundant automorphisms seems to be reducing with the list size $M$, and a small number of equivalence classes $E$ reduces the effect of the AE decoder. 
Performance of 5G polar code decoded under CA-SCL with $L=4$ is plotted as a reference; proposed polar code under AE decoding outperforms CA-SCL 5G polar code for BLER$\geq 10^{-3}$, while 5G polar code shows better performance for lower BLER, however at a larger decoding cost. 

\begin{figure}[t]
  \centering
  \includegraphics[width=0.95\columnwidth]{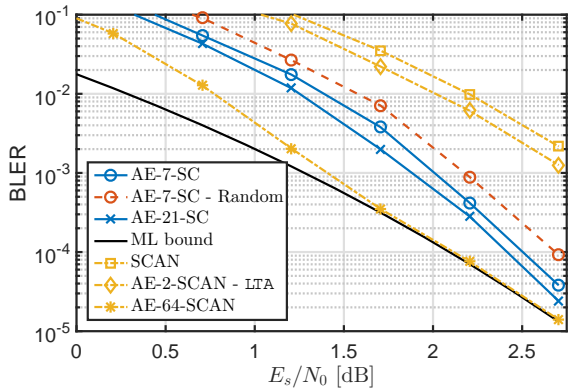}
\caption{Error-correction performance of $(256,95)$ polar codes under AE-SC and AE-SCAN decoding.}
\label{fig:256_95}
\end{figure}
Next, in Figure~\ref{fig:256_95} we investigate the error-correction performance of the $(256,95)$ code defined by $\mathcal{I}_{min}=\{55,120,228\}$ and having $\mathcal{A}=\BLTA(2,1,1,1,3)$.
The number of $\EC$s under SC decoding for this code can be calculated as $E=21$ by knowing that $[\pid]=\BLTA(2,1,1,1,1,1,1)$.
For $M=7$, AE-SC decoding with a random automorphism set suffers from a loss with respect to AE-SC decoding performed with automorphisms from $7$ distinct $\EC$s.
The bound under AE-SC decoding is obtained by using $M=21$ non-redundant automorphisms, one representative from each $\EC$.
However, this AE bound is far away from the ML bound of the code; then, the error-correction performance under AE-SCAN decoding is analyzed.
We observe that error-correction performance of AE-$2$-SCAN with a set composed of two $\LTA$ automorphisms is not equivalent to the error-correction performance of SCAN.
Thus, the absorption group of SCAN is smaller than $\LTA$, permitting to use additional automorphisms without redundancy; the characterization of the absorption set under SCAN decoding, that we conjecture to be limited to the trivial permutation, is still an open problem. 
Decoding performance of AE-$64$-SCAN designed with a set of $64$ random automorphism from $\BLTA(2,1,1,1,3)$ matches the ML bound for low BLER.

\begin{figure}[t]
  \centering
  \includegraphics[width=0.95\columnwidth]{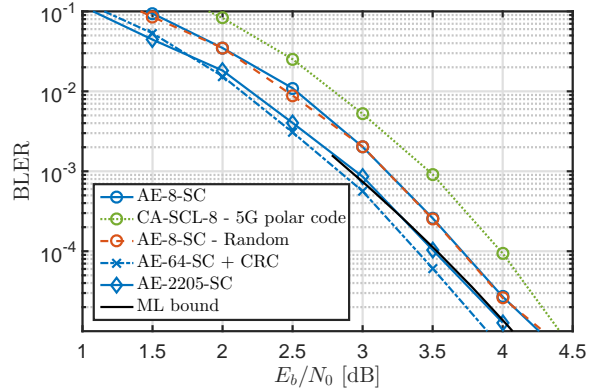}
\caption{Error-correction performance of $(128,60)$ polar code with $\mathcal{A}=\BLTA(3,4)$  under AE-SC.}
\label{fig:128_60}
\end{figure}
Figure~\ref{fig:128_60} shows the error-correction of $(128,60)$ code defined by $\mathcal{I}_{min}=\{27\}$.
This code has $\mathcal{A}=\BLTA(3,4)$ and SC absorption group $[\pid]_{\scdec}=\BLTA(2,1,1,1,1,1)$.
Under SC decoding, the code exhibits $E=2205$ $\EC$s permitting to match the ML bound under AE-SC when $M=2205$ non-redundant automorphisms are used. 
Given the large number of ECs, the probability for a random set of automorphisms to include redundant permutation is quite small; as an example, $P_{\leq1}(8)=0.0126$. 
A more accurate analysis shows that for large sets $\mathcal{L}$, some $\ascdec$ decoding units returns the correct codeword, whereas AE-SC decoder selects another codeword based on the least-square metric. 
To overcome this problem, we introduced a short CRC of $3$ bits; simulation results show that in this way it is possible to beat the ML bound of the code with only $M=64$ automorphisms. 
However, the introduction of a CRC is not always useful; in fact, its introduction for the previously analyzed codes did not provide any benefit. 
How and when to introduce a CRC under AE-SC decoding is still an open problem. 
Finally, we note that the code investigated has a good ML bound allowing the error correction performance of AE decoding with $M=8$ to outperform 5G polar code decoded under CA-SCL-$8$ decoding. 


\section{Conclusions}
\label{sec:conclusions}
In this paper, we introduced the notion of redundant automorphisms of polar codes under AE decoding. 
This notion permits to greatly reduce the number of automorphisms that can be used in AE-SC decoding of polar codes. 
Moreover, by analyzing the number of non-redundant automorphisms it is possible to have an idea of the impact of an AE-SC dedoder: in fact, even if a polar code has a large affine automorphism group, if it SC absorption group is too large, the number of distinct codeword candidate under AE-SC decoding may be too small. 
Then, we introduced a method to generate a set if non-redundant automorphisms to be used in AE-SC decoding. 
All these results were made possible by a preliminar analysis of the structure of the affine automorphism group of polar codes. 
We proposed a novel approach to prove the most recent results in this field, and we provide a proof of the equivalence between decreasing monomial codes and polar codes following UPO. 
Simulation results show the goodness of our approach, however leaving open various problems: the structure of the absorption group of decoding algorithms other than SC, e.g. BP or SCAN, is still unknown, while it is not clear if it is possible to further reduce the automorphism set size by eliminating automorphisms providing the same result under a given received signal $y$. 
With our paper, we hope to provide new tools to help the researcher to answer to these and to other question related to AE decoding of polar codes, a really promising decoding algorithm for high parallel implementations. 

\bibliographystyle{IEEEbib}
\bibliography{references}

\appendix

\section*{Boolean functions and permutations}
\label{app:eval}
\begin{table}[ht]
\centering
\caption{Truth table for boolean function $f$ in \eqref{eq:f_ex}.} 
\label{tab:truth}
	\begin{tabular}{|c|c c c|c|c|c|}
		\hline
		$i$ & & $\bin{i}$ & & $\eval(\bar{V}_0)$ & $\eval(\bar{V}_1)$ & $\eval(\bar{V}_2)$ \\
		\hline
		0 &  0  &  0  &  0  & 1 & 1 & 1 \\
		1 &  1  &  0  &  0  & 0 & 1 & 1 \\
		2 &  0  &  1  &  0  & 1 & 0 & 1 \\
		3 &  1  &  1  &  0  & 0 & 0 & 1 \\
		4 &  0  &  0  &  1  & 1 & 1 & 0 \\
		5 &  1  &  0  &  1  & 0 & 1 & 0 \\
		6 &  0  &  1  &  1  & 1 & 0 & 0 \\
		7 &  1  &  1  &  1  & 0 & 0 & 0 \\
		\hline
	\end{tabular}
\end{table}
Boolean function $f \in \Fbin^{[3]}$ defined by
\begin{equation}
\label{eq:f_ex}
f = \bar{V}_2 \oplus \bar{V}_0\bar{V}_1 \oplus \bar{V}_0\bar{V}_2 \oplus \bar{V}_1\bar{V}_2 = m_{3} + m_{4} + m_{2} + m_{1}
\end{equation}
can be easily evaluated on the basis of single variable evaluations given by  
\begin{align*}
	x^{(m_6)} = \eval(\bar{V}_0) = & 10101010 \\
	x^{(m_5)} = \eval(\bar{V}_1) = & 11001100 \\
	x^{(m_3)} = \eval(\bar{V}_2) = & 11110000 \\
\end{align*}
Each single variable evaluation is described by a column of the Truth Table~\ref{tab:truth}. 
In fact, the evaluation of function $f$ is given by
\begin{center}
	\begin{tabular}{r|c c}
		$x^{(m_3)} = \eval(\bar{V}_2)$    & $11110000$ & $\oplus$ \\ 
		$x^{(m_4)} = \eval(\bar{V}_0 \bar{V}_1)$ & $10001000$ & $\oplus$ \\ 
		$x^{(m_2)} = \eval(\bar{V}_0 \bar{V}_2)$ & $10100000$ & $\oplus$ \\ 
		$x^{(m_1)} = \eval(\bar{V}_1 \bar{V}_2)$ & $11000000$ & $=$ \\ 
		\hline
		$x^{(f)} = \eval(f)$      & $00011000$
	\end{tabular}
\end{center}
Now let us apply left shift permutation $\sigma_j$ to this vector, where $\sigma_j(i) = [i+j]_{2^n}$, with $[\cdot]_{2^n}$ representing modulo $2^n$ operation; according to our notation, element in $i$ is replaced by element in $[i+j]_{2^n}$. 
Double left shift $\sigma_2$ can be expressed as an affine transformation, while $\sigma_3$ cannot; for $n=3$, $\sigma_2$ can be expressed as an affine transformation defined by
\begin{align*}
	A &= \left[ \begin{smallmatrix} 1 & 0 & 0 \\ 0 & 1 & 0 \\ 0 & 1 & 1 \end{smallmatrix} \right], &
	b &= \left[ \begin{smallmatrix} 0 \\ 1 \\ 0  \end{smallmatrix} \right]. 
\end{align*}  
Let us take boolean function $g \in \Fbin^{[3]}$ defined by $g = \bar{V}_0\bar{V}_2 \oplus \bar{V}_1\bar{V}_2 = m_{2} + m_{1}$, having evaluation $x^{(g)} = 01100000$. 
According to Lemma~\ref{lem:GA-perm}, we have that $\sigma_2\left( x^{(f)} \right) = x^{(g)}$ when $g = \sigma_2(f)$. 
In fact, according to its affine transformation definition, $\sigma_2$ maps 
\begin{align*}
	\bar{V}_0 \mapsto \bar{V}_0 , & \bar{V}_1 \mapsto \bar{V}_1 \oplus 1 , & \bar{V}_2 \mapsto \bar{V}_1 \oplus \bar{V}_2 \oplus 1. 
\end{align*} 
It is worth noticing that the use of negative monomials adds a "$\oplus 1$" addend compared to positive monomials when the number of output variables is even. 
It is easy to verify that applying this variables substitutions to boolean function $f$ leads to boolean function $g$.  

\section*{$\EC$ representatives calculation}
\label{app:rep_list}

Here we provide an example of the generation of automorphisms list $\mathcal{L}$ according to the proposed method. 
We analyze the $(256,95)$ code defined by $\mathcal{I}_{min}=\{55,120,228\}$ and having $\mathcal{A}=\BLTA(2,1,1,1,3)$, whose error correction performance are depicted in Figure~\ref{fig:256_95}. 
According to Lemma~\ref{lem:sizeBLTA}, this code has
\begin{equation}
|\mathcal{A}| = 2^{\frac{8(8+1)}{2}}\cdot \prod_{i=1}^5 \left( \prod_{j=2}^{s_i} \left( 2^j-1 \right) \right) = 2^{36} \cdot (3) \cdot (3 \cdot 7) = 63 \cdot 2^{36}
\end{equation}
affine automorphisms. 
This code has SC absorption group $[\pid]=\BLTA(2,1,1,1,1,1,1)$, and again for Lemma~\ref{lem:sizeBLTA} the number of SC absorbed affine automorphisms is given by
\begin{equation}
|[\pid]| = 2^{\frac{8(8+1)}{2}}\cdot \prod_{i=1}^7 \left( \prod_{j=2}^{s_i} \left( 2^j-1 \right) \right) = 3 \cdot 2^{36}.
\end{equation}
According to Lemma~\ref{lem:number_EC}, there are $E=\frac{|\mathcal{A}|}{|[\pid]|} = 21$ equivalence classes under SC decoding, as confirmed by Lemma~\ref{lem:size_EC_SC}. 
So, there are only $21$ possible different outcomes under AE-SC decoding, greatly reducing the maximum size $M$ of the AE decoder.

In the following, we use the result of Lemma~\ref{lem:PU_BLT} to create a list of $21$ affine automorphisms belonging to different equivalence classes under SC decoding.
In this case, there are $|\mathcal{A}_{\mathcal{U}}| = 16$ $\UTL$ automorphisms and $|\mathcal{A}_{\mathcal{P}}|=12$ $\PL$ automorphisms; it is worth noticing that $|\mathcal{A}_{\mathcal{U}}|+|\mathcal{A}_{\mathcal{P}}|>21$ since $[\pid] \subsetneq \LTA$, so automorphisms in $\mathcal{A}_{\mathcal{U}}$ and $\mathcal{A}_{\mathcal{P}}$ are not all included in $[\pid]$. 
To begin with, we generate all the $16$ UTL automorphisms composing $\mathcal{A}_{\mathcal{U}}$; each of them has the form $\pi_{(U,0)}$ where
\begin{equation}
U = \begin{bmatrix} U_1 & & & & 0 \\ & 1 & & & \\ & & 1 & & \\& & & 1 & \\ 0 & & & & U_5\\ \end{bmatrix},
\end{equation}
and $U_1$ is a $2 \times 2$ $\UTL$ matrix, while $U_5$ is a $3 \times 3$ $\UTL$ matrix. 
There are only $2$ $\UTL$ matrices of size 2 and $8$ $\UTL$ matrices of size $3$; combining them in $U$ it is possible to list all the $16$ $\UTL$ affine automorphisms of the code. 
Similarly, a $\PL$ automorphism of the code is described by where
\begin{equation}
P = \begin{bmatrix} P_1 & & & & 0 \\ & 1 & & & \\ & & 1 & & \\& & & 1 & \\ 0 & & & & P_5\\ \end{bmatrix}
\end{equation}
and $P_1$ is a $2 \times 2$ $\PL$ matrix and $P_5$ is a $3 \times 3$ $\PL$ matrix. 
Again, there are only $2$ $\PL$ matrices of size 2 and $6$ $\PL$ matrices of size $3$, for a total of $12$ $\PL$ automorphisms. 
In order to generate $\mathcal{L} = \{\pi_i,\ldots,\pi_{21}\}$, where $\pi_i = \pi_{{A_i,0}}$, we begin by setting $A_1 = I$, such that $[\pi_1] = [\pid]$. 
Next, we star adding to $\mathcal{L}$ the elements of $\mathcal{A}_{\mathcal{U}}$; only $7$ of them can be added, an they are listed as $A_2,\ldots,A_8$. 
This happens because the first block $U_1$ of an $\UTL$ matrix $U$ is always included in the absorption group, and hence only the matrices having different last block $U_5$ can be included, as stated by Lemma~\ref{lem:checkproductA}. 
Next, elements of $\mathcal{A}_{\mathcal{P}}$ are added; only $5$ of them are independent from the already included ones, and they are listed as $A_9,\ldots,A_{13}$. 
Finally, the remaining elements of $\mathcal{L}$ need to be calculated as $P \cdot U$, with $U \in \mathcal{A}_{\mathcal{U}}$ and $P \in \mathcal{A}_{\mathcal{P}}$, as stated in Lemma~\ref{lem:PU_BLT}. 
There are $|\mathcal{A}_{\mathcal{U}}| \cdot |\mathcal{A}_{\mathcal{P}}| = 192$ possible combinations of $\UTL$ an $\PL$ automorphisms; however, it is not mandatory to calculate all of them, since when the remaining $8$ independent matrices are calculated, the process can stop. 
Finally, we can generate the automorphism list $\mathcal{L} = \{\pi_i,\ldots,\pi_{21}\}$, where $\pi_i = \pi_{{A_i,0}}$ and 
\ifdefined\COLS
	\begin{align*}
	A_1 &= \left[ \begin{smallmatrix} 
			1 & 0 & 0 & 0 & 0 & 0 & 0 & 0 \\ 
			0 & 1 & 0 & 0 & 0 & 0 & 0 & 0 \\
			0 & 0 & 1 & 0 & 0 & 0 & 0 & 0 \\
			0 & 0 & 0 & 1 & 0 & 0 & 0 & 0 \\
			0 & 0 & 0 & 0 & 1 & 0 & 0 & 0 \\
			0 & 0 & 0 & 0 & 0 & 1 & 0 & 0 \\
			0 & 0 & 0 & 0 & 0 & 0 & 1 & 0 \\
			0 & 0 & 0 & 0 & 0 & 0 & 0 & 1 
			\end{smallmatrix} \right], &
	A_2 &= \left[ \begin{smallmatrix} 
			1 & 0 & 0 & 0 & 0 & 0 & 0 & 0 \\ 
			0 & 1 & 0 & 0 & 0 & 0 & 0 & 0 \\
			0 & 0 & 1 & 0 & 0 & 0 & 0 & 0 \\
			0 & 0 & 0 & 1 & 0 & 0 & 0 & 0 \\
			0 & 0 & 0 & 0 & 1 & 0 & 0 & 0 \\
			0 & 0 & 0 & 0 & 0 & 1 & 0 & 0 \\
			0 & 0 & 0 & 0 & 0 & 0 & 1 & 1 \\
			0 & 0 & 0 & 0 & 0 & 0 & 0 & 1 
			\end{smallmatrix} \right], \\
	A_3 &= \left[ \begin{smallmatrix} 
			1 & 0 & 0 & 0 & 0 & 0 & 0 & 0 \\ 
			0 & 1 & 0 & 0 & 0 & 0 & 0 & 0 \\
			0 & 0 & 1 & 0 & 0 & 0 & 0 & 0 \\
			0 & 0 & 0 & 1 & 0 & 0 & 0 & 0 \\
			0 & 0 & 0 & 0 & 1 & 0 & 0 & 0 \\
			0 & 0 & 0 & 0 & 0 & 1 & 0 & 1 \\
			0 & 0 & 0 & 0 & 0 & 0 & 1 & 0 \\
			0 & 0 & 0 & 0 & 0 & 0 & 0 & 1 
			\end{smallmatrix} \right], &
	A_4 &= \left[ \begin{smallmatrix} 
			1 & 0 & 0 & 0 & 0 & 0 & 0 & 0 \\ 
			0 & 1 & 0 & 0 & 0 & 0 & 0 & 0 \\
			0 & 0 & 1 & 0 & 0 & 0 & 0 & 0 \\
			0 & 0 & 0 & 1 & 0 & 0 & 0 & 0 \\
			0 & 0 & 0 & 0 & 1 & 0 & 0 & 0 \\
			0 & 0 & 0 & 0 & 0 & 1 & 0 & 1 \\
			0 & 0 & 0 & 0 & 0 & 0 & 1 & 1 \\
			0 & 0 & 0 & 0 & 0 & 0 & 0 & 1 
			\end{smallmatrix} \right], \\
	A_5 &= \left[ \begin{smallmatrix} 
			1 & 0 & 0 & 0 & 0 & 0 & 0 & 0 \\ 
			0 & 1 & 0 & 0 & 0 & 0 & 0 & 0 \\
			0 & 0 & 1 & 0 & 0 & 0 & 0 & 0 \\
			0 & 0 & 0 & 1 & 0 & 0 & 0 & 0 \\
			0 & 0 & 0 & 0 & 1 & 0 & 0 & 0 \\
			0 & 0 & 0 & 0 & 0 & 1 & 1 & 0 \\
			0 & 0 & 0 & 0 & 0 & 0 & 1 & 0 \\
			0 & 0 & 0 & 0 & 0 & 0 & 0 & 1 
			\end{smallmatrix} \right], &
	A_6 &= \left[ \begin{smallmatrix} 
			1 & 0 & 0 & 0 & 0 & 0 & 0 & 0 \\ 
			0 & 1 & 0 & 0 & 0 & 0 & 0 & 0 \\
			0 & 0 & 1 & 0 & 0 & 0 & 0 & 0 \\
			0 & 0 & 0 & 1 & 0 & 0 & 0 & 0 \\
			0 & 0 & 0 & 0 & 1 & 0 & 0 & 0 \\
			0 & 0 & 0 & 0 & 0 & 1 & 1 & 0 \\
			0 & 0 & 0 & 0 & 0 & 0 & 1 & 1 \\
			0 & 0 & 0 & 0 & 0 & 0 & 0 & 1 
			\end{smallmatrix} \right], \\
	A_7 &= \left[ \begin{smallmatrix} 
			1 & 0 & 0 & 0 & 0 & 0 & 0 & 0 \\ 
			0 & 1 & 0 & 0 & 0 & 0 & 0 & 0 \\
			0 & 0 & 1 & 0 & 0 & 0 & 0 & 0 \\
			0 & 0 & 0 & 1 & 0 & 0 & 0 & 0 \\
			0 & 0 & 0 & 0 & 1 & 0 & 0 & 0 \\
			0 & 0 & 0 & 0 & 0 & 1 & 1 & 1 \\
			0 & 0 & 0 & 0 & 0 & 0 & 1 & 0 \\
			0 & 0 & 0 & 0 & 0 & 0 & 0 & 1 
			\end{smallmatrix} \right], &
	A_8 &= \left[ \begin{smallmatrix} 
			1 & 0 & 0 & 0 & 0 & 0 & 0 & 0 \\ 
			0 & 1 & 0 & 0 & 0 & 0 & 0 & 0 \\
			0 & 0 & 1 & 0 & 0 & 0 & 0 & 0 \\
			0 & 0 & 0 & 1 & 0 & 0 & 0 & 0 \\
			0 & 0 & 0 & 0 & 1 & 0 & 0 & 0 \\
			0 & 0 & 0 & 0 & 0 & 1 & 1 & 1 \\
			0 & 0 & 0 & 0 & 0 & 0 & 1 & 1 \\
			0 & 0 & 0 & 0 & 0 & 0 & 0 & 1 
			\end{smallmatrix} \right], \\
	A_9 &= \left[ \begin{smallmatrix} 
			1 & 0 & 0 & 0 & 0 & 0 & 0 & 0 \\ 
			0 & 1 & 0 & 0 & 0 & 0 & 0 & 0 \\
			0 & 0 & 1 & 0 & 0 & 0 & 0 & 0 \\
			0 & 0 & 0 & 1 & 0 & 0 & 0 & 0 \\
			0 & 0 & 0 & 0 & 1 & 0 & 0 & 0 \\
			0 & 0 & 0 & 0 & 0 & 0 & 0 & 1 \\
			0 & 0 & 0 & 0 & 0 & 1 & 0 & 0 \\
			0 & 0 & 0 & 0 & 0 & 0 & 1 & 0 
			\end{smallmatrix} \right], &
	A_{10} &= \left[ \begin{smallmatrix} 
			1 & 0 & 0 & 0 & 0 & 0 & 0 & 0 \\ 
			0 & 1 & 0 & 0 & 0 & 0 & 0 & 0 \\
			0 & 0 & 1 & 0 & 0 & 0 & 0 & 0 \\
			0 & 0 & 0 & 1 & 0 & 0 & 0 & 0 \\
			0 & 0 & 0 & 0 & 1 & 0 & 0 & 0 \\
			0 & 0 & 0 & 0 & 0 & 0 & 1 & 0 \\
			0 & 0 & 0 & 0 & 0 & 0 & 0 & 1 \\
			0 & 0 & 0 & 0 & 0 & 1 & 0 & 0 
			\end{smallmatrix} \right], \\
	A_{11} &= \left[ \begin{smallmatrix} 
			1 & 0 & 0 & 0 & 0 & 0 & 0 & 0 \\ 
			0 & 1 & 0 & 0 & 0 & 0 & 0 & 0 \\
			0 & 0 & 1 & 0 & 0 & 0 & 0 & 0 \\
			0 & 0 & 0 & 1 & 0 & 0 & 0 & 0 \\
			0 & 0 & 0 & 0 & 1 & 0 & 0 & 0 \\
			0 & 0 & 0 & 0 & 0 & 0 & 1 & 0 \\
			0 & 0 & 0 & 0 & 0 & 1 & 0 & 0 \\
			0 & 0 & 0 & 0 & 0 & 0 & 0 & 1 
			\end{smallmatrix} \right], &
	A_{12} &= \left[ \begin{smallmatrix} 
			1 & 0 & 0 & 0 & 0 & 0 & 0 & 0 \\ 
			0 & 1 & 0 & 0 & 0 & 0 & 0 & 0 \\
			0 & 0 & 1 & 0 & 0 & 0 & 0 & 0 \\
			0 & 0 & 0 & 1 & 0 & 0 & 0 & 0 \\
			0 & 0 & 0 & 0 & 1 & 0 & 0 & 0 \\
			0 & 0 & 0 & 0 & 0 & 1 & 0 & 0 \\
			0 & 0 & 0 & 0 & 0 & 0 & 0 & 1 \\
			0 & 0 & 0 & 0 & 0 & 0 & 1 & 0 
			\end{smallmatrix} \right], \\
	A_{13} &= \left[ \begin{smallmatrix} 
			1 & 0 & 0 & 0 & 0 & 0 & 0 & 0 \\ 
			0 & 1 & 0 & 0 & 0 & 0 & 0 & 0 \\
			0 & 0 & 1 & 0 & 0 & 0 & 0 & 0 \\
			0 & 0 & 0 & 1 & 0 & 0 & 0 & 0 \\
			0 & 0 & 0 & 0 & 1 & 0 & 0 & 0 \\
			0 & 0 & 0 & 0 & 0 & 0 & 0 & 1 \\
			0 & 0 & 0 & 0 & 0 & 0 & 1 & 0 \\
			0 & 0 & 0 & 0 & 0 & 1 & 0 & 0 
			\end{smallmatrix} \right], &
	A_{14} &= \left[ \begin{smallmatrix} 
			1 & 0 & 0 & 0 & 0 & 0 & 0 & 0 \\ 
			0 & 1 & 0 & 0 & 0 & 0 & 0 & 0 \\
			0 & 0 & 1 & 0 & 0 & 0 & 0 & 0 \\
			0 & 0 & 0 & 1 & 0 & 0 & 0 & 0 \\
			0 & 0 & 0 & 0 & 1 & 0 & 0 & 0 \\
			0 & 0 & 0 & 0 & 0 & 0 & 1 & 1 \\
			0 & 0 & 0 & 0 & 0 & 0 & 0 & 1 \\
			0 & 0 & 0 & 0 & 0 & 1 & 0 & 0 
			\end{smallmatrix} \right], \\
	A_{15} &= \left[ \begin{smallmatrix} 
			1 & 0 & 0 & 0 & 0 & 0 & 0 & 0 \\ 
			0 & 1 & 0 & 0 & 0 & 0 & 0 & 0 \\
			0 & 0 & 1 & 0 & 0 & 0 & 0 & 0 \\
			0 & 0 & 0 & 1 & 0 & 0 & 0 & 0 \\
			0 & 0 & 0 & 0 & 1 & 0 & 0 & 0 \\
			0 & 0 & 0 & 0 & 0 & 0 & 1 & 1 \\
			0 & 0 & 0 & 0 & 0 & 1 & 0 & 0 \\
			0 & 0 & 0 & 0 & 0 & 0 & 0 & 1 
			\end{smallmatrix} \right], &
	A_{16} &= \left[ \begin{smallmatrix} 
			1 & 0 & 0 & 0 & 0 & 0 & 0 & 0 \\ 
			0 & 1 & 0 & 0 & 0 & 0 & 0 & 0 \\
			0 & 0 & 1 & 0 & 0 & 0 & 0 & 0 \\
			0 & 0 & 0 & 1 & 0 & 0 & 0 & 0 \\
			0 & 0 & 0 & 0 & 1 & 0 & 0 & 0 \\
			0 & 0 & 0 & 0 & 0 & 0 & 1 & 0 \\
			0 & 0 & 0 & 0 & 0 & 1 & 0 & 1 \\
			0 & 0 & 0 & 0 & 0 & 0 & 0 & 1 
			\end{smallmatrix} \right], \\
	A_{17} &= \left[ \begin{smallmatrix} 
			1 & 0 & 0 & 0 & 0 & 0 & 0 & 0 \\ 
			0 & 1 & 0 & 0 & 0 & 0 & 0 & 0 \\
			0 & 0 & 1 & 0 & 0 & 0 & 0 & 0 \\
			0 & 0 & 0 & 1 & 0 & 0 & 0 & 0 \\
			0 & 0 & 0 & 0 & 1 & 0 & 0 & 0 \\
			0 & 0 & 0 & 0 & 0 & 1 & 0 & 1 \\
			0 & 0 & 0 & 0 & 0 & 0 & 0 & 1 \\
			0 & 0 & 0 & 0 & 0 & 0 & 1 & 0 
			\end{smallmatrix} \right], &
	A_{18} &= \left[ \begin{smallmatrix} 
			1 & 0 & 0 & 0 & 0 & 0 & 0 & 0 \\ 
			0 & 1 & 0 & 0 & 0 & 0 & 0 & 0 \\
			0 & 0 & 1 & 0 & 0 & 0 & 0 & 0 \\
			0 & 0 & 0 & 1 & 0 & 0 & 0 & 0 \\
			0 & 0 & 0 & 0 & 1 & 0 & 0 & 0 \\
			0 & 0 & 0 & 0 & 0 & 0 & 1 & 1 \\
			0 & 0 & 0 & 0 & 0 & 1 & 0 & 1 \\
			0 & 0 & 0 & 0 & 0 & 0 & 0 & 1 
			\end{smallmatrix} \right], \\
	A_{19} &= \left[ \begin{smallmatrix} 
			1 & 0 & 0 & 0 & 0 & 0 & 0 & 0 \\ 
			0 & 1 & 0 & 0 & 0 & 0 & 0 & 0 \\
			0 & 0 & 1 & 0 & 0 & 0 & 0 & 0 \\
			0 & 0 & 0 & 1 & 0 & 0 & 0 & 0 \\
			0 & 0 & 0 & 0 & 1 & 0 & 0 & 0 \\
			0 & 0 & 0 & 0 & 0 & 0 & 0 & 1 \\
			0 & 0 & 0 & 0 & 0 & 1 & 1 & 0 \\
			0 & 0 & 0 & 0 & 0 & 0 & 1 & 0 
			\end{smallmatrix} \right], &
	A_{20} &= \left[ \begin{smallmatrix} 
			1 & 0 & 0 & 0 & 0 & 0 & 0 & 0 \\ 
			0 & 1 & 0 & 0 & 0 & 0 & 0 & 0 \\
			0 & 0 & 1 & 0 & 0 & 0 & 0 & 0 \\
			0 & 0 & 0 & 1 & 0 & 0 & 0 & 0 \\
			0 & 0 & 0 & 0 & 1 & 0 & 0 & 0 \\
			0 & 0 & 0 & 0 & 0 & 1 & 1 & 0 \\
			0 & 0 & 0 & 0 & 0 & 0 & 0 & 1 \\
			0 & 0 & 0 & 0 & 0 & 0 & 1 & 0 
			\end{smallmatrix} \right], \\
	A_{21} &= \left[ \begin{smallmatrix} 
			1 & 0 & 0 & 0 & 0 & 0 & 0 & 0 \\ 
			0 & 1 & 0 & 0 & 0 & 0 & 0 & 0 \\
			0 & 0 & 1 & 0 & 0 & 0 & 0 & 0 \\
			0 & 0 & 0 & 1 & 0 & 0 & 0 & 0 \\
			0 & 0 & 0 & 0 & 1 & 0 & 0 & 0 \\
			0 & 0 & 0 & 0 & 0 & 1 & 1 & 1 \\
			0 & 0 & 0 & 0 & 0 & 0 & 0 & 1 \\
			0 & 0 & 0 & 0 & 0 & 0 & 1 & 0 
			\end{smallmatrix} \right].											
\end{align*}
\else
	\begin{align*}
	A_1 &= \left[ \begin{smallmatrix} 
			1 & 0 & 0 & 0 & 0 & 0 & 0 & 0 \\ 
			0 & 1 & 0 & 0 & 0 & 0 & 0 & 0 \\
			0 & 0 & 1 & 0 & 0 & 0 & 0 & 0 \\
			0 & 0 & 0 & 1 & 0 & 0 & 0 & 0 \\
			0 & 0 & 0 & 0 & 1 & 0 & 0 & 0 \\
			0 & 0 & 0 & 0 & 0 & 1 & 0 & 0 \\
			0 & 0 & 0 & 0 & 0 & 0 & 1 & 0 \\
			0 & 0 & 0 & 0 & 0 & 0 & 0 & 1 
			\end{smallmatrix} \right], &
	A_2 &= \left[ \begin{smallmatrix} 
			1 & 0 & 0 & 0 & 0 & 0 & 0 & 0 \\ 
			0 & 1 & 0 & 0 & 0 & 0 & 0 & 0 \\
			0 & 0 & 1 & 0 & 0 & 0 & 0 & 0 \\
			0 & 0 & 0 & 1 & 0 & 0 & 0 & 0 \\
			0 & 0 & 0 & 0 & 1 & 0 & 0 & 0 \\
			0 & 0 & 0 & 0 & 0 & 1 & 0 & 0 \\
			0 & 0 & 0 & 0 & 0 & 0 & 1 & 1 \\
			0 & 0 & 0 & 0 & 0 & 0 & 0 & 1 
			\end{smallmatrix} \right], &
	A_3 &= \left[ \begin{smallmatrix} 
			1 & 0 & 0 & 0 & 0 & 0 & 0 & 0 \\ 
			0 & 1 & 0 & 0 & 0 & 0 & 0 & 0 \\
			0 & 0 & 1 & 0 & 0 & 0 & 0 & 0 \\
			0 & 0 & 0 & 1 & 0 & 0 & 0 & 0 \\
			0 & 0 & 0 & 0 & 1 & 0 & 0 & 0 \\
			0 & 0 & 0 & 0 & 0 & 1 & 0 & 1 \\
			0 & 0 & 0 & 0 & 0 & 0 & 1 & 0 \\
			0 & 0 & 0 & 0 & 0 & 0 & 0 & 1 
			\end{smallmatrix} \right], \\
	A_4 &= \left[ \begin{smallmatrix} 
			1 & 0 & 0 & 0 & 0 & 0 & 0 & 0 \\ 
			0 & 1 & 0 & 0 & 0 & 0 & 0 & 0 \\
			0 & 0 & 1 & 0 & 0 & 0 & 0 & 0 \\
			0 & 0 & 0 & 1 & 0 & 0 & 0 & 0 \\
			0 & 0 & 0 & 0 & 1 & 0 & 0 & 0 \\
			0 & 0 & 0 & 0 & 0 & 1 & 0 & 1 \\
			0 & 0 & 0 & 0 & 0 & 0 & 1 & 1 \\
			0 & 0 & 0 & 0 & 0 & 0 & 0 & 1 
			\end{smallmatrix} \right], &
	A_5 &= \left[ \begin{smallmatrix} 
			1 & 0 & 0 & 0 & 0 & 0 & 0 & 0 \\ 
			0 & 1 & 0 & 0 & 0 & 0 & 0 & 0 \\
			0 & 0 & 1 & 0 & 0 & 0 & 0 & 0 \\
			0 & 0 & 0 & 1 & 0 & 0 & 0 & 0 \\
			0 & 0 & 0 & 0 & 1 & 0 & 0 & 0 \\
			0 & 0 & 0 & 0 & 0 & 1 & 1 & 0 \\
			0 & 0 & 0 & 0 & 0 & 0 & 1 & 0 \\
			0 & 0 & 0 & 0 & 0 & 0 & 0 & 1 
			\end{smallmatrix} \right], &
	A_6 &= \left[ \begin{smallmatrix} 
			1 & 0 & 0 & 0 & 0 & 0 & 0 & 0 \\ 
			0 & 1 & 0 & 0 & 0 & 0 & 0 & 0 \\
			0 & 0 & 1 & 0 & 0 & 0 & 0 & 0 \\
			0 & 0 & 0 & 1 & 0 & 0 & 0 & 0 \\
			0 & 0 & 0 & 0 & 1 & 0 & 0 & 0 \\
			0 & 0 & 0 & 0 & 0 & 1 & 1 & 0 \\
			0 & 0 & 0 & 0 & 0 & 0 & 1 & 1 \\
			0 & 0 & 0 & 0 & 0 & 0 & 0 & 1 
			\end{smallmatrix} \right], \\
	A_7 &= \left[ \begin{smallmatrix} 
			1 & 0 & 0 & 0 & 0 & 0 & 0 & 0 \\ 
			0 & 1 & 0 & 0 & 0 & 0 & 0 & 0 \\
			0 & 0 & 1 & 0 & 0 & 0 & 0 & 0 \\
			0 & 0 & 0 & 1 & 0 & 0 & 0 & 0 \\
			0 & 0 & 0 & 0 & 1 & 0 & 0 & 0 \\
			0 & 0 & 0 & 0 & 0 & 1 & 1 & 1 \\
			0 & 0 & 0 & 0 & 0 & 0 & 1 & 0 \\
			0 & 0 & 0 & 0 & 0 & 0 & 0 & 1 
			\end{smallmatrix} \right], &
	A_8 &= \left[ \begin{smallmatrix} 
			1 & 0 & 0 & 0 & 0 & 0 & 0 & 0 \\ 
			0 & 1 & 0 & 0 & 0 & 0 & 0 & 0 \\
			0 & 0 & 1 & 0 & 0 & 0 & 0 & 0 \\
			0 & 0 & 0 & 1 & 0 & 0 & 0 & 0 \\
			0 & 0 & 0 & 0 & 1 & 0 & 0 & 0 \\
			0 & 0 & 0 & 0 & 0 & 1 & 1 & 1 \\
			0 & 0 & 0 & 0 & 0 & 0 & 1 & 1 \\
			0 & 0 & 0 & 0 & 0 & 0 & 0 & 1 
			\end{smallmatrix} \right], &
	A_9 &= \left[ \begin{smallmatrix} 
			1 & 0 & 0 & 0 & 0 & 0 & 0 & 0 \\ 
			0 & 1 & 0 & 0 & 0 & 0 & 0 & 0 \\
			0 & 0 & 1 & 0 & 0 & 0 & 0 & 0 \\
			0 & 0 & 0 & 1 & 0 & 0 & 0 & 0 \\
			0 & 0 & 0 & 0 & 1 & 0 & 0 & 0 \\
			0 & 0 & 0 & 0 & 0 & 0 & 0 & 1 \\
			0 & 0 & 0 & 0 & 0 & 1 & 0 & 0 \\
			0 & 0 & 0 & 0 & 0 & 0 & 1 & 0 
			\end{smallmatrix} \right], \\
	A_{10} &= \left[ \begin{smallmatrix} 
			1 & 0 & 0 & 0 & 0 & 0 & 0 & 0 \\ 
			0 & 1 & 0 & 0 & 0 & 0 & 0 & 0 \\
			0 & 0 & 1 & 0 & 0 & 0 & 0 & 0 \\
			0 & 0 & 0 & 1 & 0 & 0 & 0 & 0 \\
			0 & 0 & 0 & 0 & 1 & 0 & 0 & 0 \\
			0 & 0 & 0 & 0 & 0 & 0 & 1 & 0 \\
			0 & 0 & 0 & 0 & 0 & 0 & 0 & 1 \\
			0 & 0 & 0 & 0 & 0 & 1 & 0 & 0 
			\end{smallmatrix} \right], &
	A_{11} &= \left[ \begin{smallmatrix} 
			1 & 0 & 0 & 0 & 0 & 0 & 0 & 0 \\ 
			0 & 1 & 0 & 0 & 0 & 0 & 0 & 0 \\
			0 & 0 & 1 & 0 & 0 & 0 & 0 & 0 \\
			0 & 0 & 0 & 1 & 0 & 0 & 0 & 0 \\
			0 & 0 & 0 & 0 & 1 & 0 & 0 & 0 \\
			0 & 0 & 0 & 0 & 0 & 0 & 1 & 0 \\
			0 & 0 & 0 & 0 & 0 & 1 & 0 & 0 \\
			0 & 0 & 0 & 0 & 0 & 0 & 0 & 1 
			\end{smallmatrix} \right], &
	A_{12} &= \left[ \begin{smallmatrix} 
			1 & 0 & 0 & 0 & 0 & 0 & 0 & 0 \\ 
			0 & 1 & 0 & 0 & 0 & 0 & 0 & 0 \\
			0 & 0 & 1 & 0 & 0 & 0 & 0 & 0 \\
			0 & 0 & 0 & 1 & 0 & 0 & 0 & 0 \\
			0 & 0 & 0 & 0 & 1 & 0 & 0 & 0 \\
			0 & 0 & 0 & 0 & 0 & 1 & 0 & 0 \\
			0 & 0 & 0 & 0 & 0 & 0 & 0 & 1 \\
			0 & 0 & 0 & 0 & 0 & 0 & 1 & 0 
			\end{smallmatrix} \right], \\
	A_{13} &= \left[ \begin{smallmatrix} 
			1 & 0 & 0 & 0 & 0 & 0 & 0 & 0 \\ 
			0 & 1 & 0 & 0 & 0 & 0 & 0 & 0 \\
			0 & 0 & 1 & 0 & 0 & 0 & 0 & 0 \\
			0 & 0 & 0 & 1 & 0 & 0 & 0 & 0 \\
			0 & 0 & 0 & 0 & 1 & 0 & 0 & 0 \\
			0 & 0 & 0 & 0 & 0 & 0 & 0 & 1 \\
			0 & 0 & 0 & 0 & 0 & 0 & 1 & 0 \\
			0 & 0 & 0 & 0 & 0 & 1 & 0 & 0 
			\end{smallmatrix} \right], &
	A_{14} &= \left[ \begin{smallmatrix} 
			1 & 0 & 0 & 0 & 0 & 0 & 0 & 0 \\ 
			0 & 1 & 0 & 0 & 0 & 0 & 0 & 0 \\
			0 & 0 & 1 & 0 & 0 & 0 & 0 & 0 \\
			0 & 0 & 0 & 1 & 0 & 0 & 0 & 0 \\
			0 & 0 & 0 & 0 & 1 & 0 & 0 & 0 \\
			0 & 0 & 0 & 0 & 0 & 0 & 1 & 1 \\
			0 & 0 & 0 & 0 & 0 & 0 & 0 & 1 \\
			0 & 0 & 0 & 0 & 0 & 1 & 0 & 0 
			\end{smallmatrix} \right], &
	A_{15} &= \left[ \begin{smallmatrix} 
			1 & 0 & 0 & 0 & 0 & 0 & 0 & 0 \\ 
			0 & 1 & 0 & 0 & 0 & 0 & 0 & 0 \\
			0 & 0 & 1 & 0 & 0 & 0 & 0 & 0 \\
			0 & 0 & 0 & 1 & 0 & 0 & 0 & 0 \\
			0 & 0 & 0 & 0 & 1 & 0 & 0 & 0 \\
			0 & 0 & 0 & 0 & 0 & 0 & 1 & 1 \\
			0 & 0 & 0 & 0 & 0 & 1 & 0 & 0 \\
			0 & 0 & 0 & 0 & 0 & 0 & 0 & 1 
			\end{smallmatrix} \right], \\
	A_{16} &= \left[ \begin{smallmatrix} 
			1 & 0 & 0 & 0 & 0 & 0 & 0 & 0 \\ 
			0 & 1 & 0 & 0 & 0 & 0 & 0 & 0 \\
			0 & 0 & 1 & 0 & 0 & 0 & 0 & 0 \\
			0 & 0 & 0 & 1 & 0 & 0 & 0 & 0 \\
			0 & 0 & 0 & 0 & 1 & 0 & 0 & 0 \\
			0 & 0 & 0 & 0 & 0 & 0 & 1 & 0 \\
			0 & 0 & 0 & 0 & 0 & 1 & 0 & 1 \\
			0 & 0 & 0 & 0 & 0 & 0 & 0 & 1 
			\end{smallmatrix} \right], &
	A_{17} &= \left[ \begin{smallmatrix} 
			1 & 0 & 0 & 0 & 0 & 0 & 0 & 0 \\ 
			0 & 1 & 0 & 0 & 0 & 0 & 0 & 0 \\
			0 & 0 & 1 & 0 & 0 & 0 & 0 & 0 \\
			0 & 0 & 0 & 1 & 0 & 0 & 0 & 0 \\
			0 & 0 & 0 & 0 & 1 & 0 & 0 & 0 \\
			0 & 0 & 0 & 0 & 0 & 1 & 0 & 1 \\
			0 & 0 & 0 & 0 & 0 & 0 & 0 & 1 \\
			0 & 0 & 0 & 0 & 0 & 0 & 1 & 0 
			\end{smallmatrix} \right], &
	A_{18} &= \left[ \begin{smallmatrix} 
			1 & 0 & 0 & 0 & 0 & 0 & 0 & 0 \\ 
			0 & 1 & 0 & 0 & 0 & 0 & 0 & 0 \\
			0 & 0 & 1 & 0 & 0 & 0 & 0 & 0 \\
			0 & 0 & 0 & 1 & 0 & 0 & 0 & 0 \\
			0 & 0 & 0 & 0 & 1 & 0 & 0 & 0 \\
			0 & 0 & 0 & 0 & 0 & 0 & 1 & 1 \\
			0 & 0 & 0 & 0 & 0 & 1 & 0 & 1 \\
			0 & 0 & 0 & 0 & 0 & 0 & 0 & 1 
			\end{smallmatrix} \right], \\
	A_{19} &= \left[ \begin{smallmatrix} 
			1 & 0 & 0 & 0 & 0 & 0 & 0 & 0 \\ 
			0 & 1 & 0 & 0 & 0 & 0 & 0 & 0 \\
			0 & 0 & 1 & 0 & 0 & 0 & 0 & 0 \\
			0 & 0 & 0 & 1 & 0 & 0 & 0 & 0 \\
			0 & 0 & 0 & 0 & 1 & 0 & 0 & 0 \\
			0 & 0 & 0 & 0 & 0 & 0 & 0 & 1 \\
			0 & 0 & 0 & 0 & 0 & 1 & 1 & 0 \\
			0 & 0 & 0 & 0 & 0 & 0 & 1 & 0 
			\end{smallmatrix} \right], &
	A_{20} &= \left[ \begin{smallmatrix} 
			1 & 0 & 0 & 0 & 0 & 0 & 0 & 0 \\ 
			0 & 1 & 0 & 0 & 0 & 0 & 0 & 0 \\
			0 & 0 & 1 & 0 & 0 & 0 & 0 & 0 \\
			0 & 0 & 0 & 1 & 0 & 0 & 0 & 0 \\
			0 & 0 & 0 & 0 & 1 & 0 & 0 & 0 \\
			0 & 0 & 0 & 0 & 0 & 1 & 1 & 0 \\
			0 & 0 & 0 & 0 & 0 & 0 & 0 & 1 \\
			0 & 0 & 0 & 0 & 0 & 0 & 1 & 0 
			\end{smallmatrix} \right], &
	A_{21} &= \left[ \begin{smallmatrix} 
			1 & 0 & 0 & 0 & 0 & 0 & 0 & 0 \\ 
			0 & 1 & 0 & 0 & 0 & 0 & 0 & 0 \\
			0 & 0 & 1 & 0 & 0 & 0 & 0 & 0 \\
			0 & 0 & 0 & 1 & 0 & 0 & 0 & 0 \\
			0 & 0 & 0 & 0 & 1 & 0 & 0 & 0 \\
			0 & 0 & 0 & 0 & 0 & 1 & 1 & 1 \\
			0 & 0 & 0 & 0 & 0 & 0 & 0 & 1 \\
			0 & 0 & 0 & 0 & 0 & 0 & 1 & 0 
			\end{smallmatrix} \right].											
	\end{align*}
\fi



\end{document}